\documentclass[sigconf]{aamas} 

\usepackage{balance} 

\setcopyright{ifaamas}
\acmConference[AAMAS '24]{Proc.\@ of the 23rd International Conference
on Autonomous Agents and Multiagent Systems (AAMAS 2024)}{May 6 -- 10, 2024}
{Auckland, New Zealand}{N.~Alechina, V.~Dignum, M.~Dastani, J.S.~Sichman (eds.)}
\copyrightyear{2024}
\acmYear{2024}
\acmDOI{}
\acmPrice{}
\acmISBN{}

\acmSubmissionID{211}

\title{Is Limited Information Enough? An Approximate Multi-agent Coverage Control in Non-Convex Discrete Environments}

\iftrue
\author{Tatsuya Iwase}
\affiliation{
  \institution{Toyota Motor Europe NV/SA}
  \city{Zaventem}
  \country{Belgium}}
\email{tiwase@mosk.tytlabs.co.jp}

\author{Aur\'{e}lie Beynier}
\affiliation{
  \institution{Sorbonne Universit\'{e}, CNRS \\ LIP6}
  \city{F-75005 Paris}
  \country{France}}
\email{aurelie.beynier@lip6.fr}

\author{Nicolas Bredeche}
\affiliation{
  \institution{Sorbonne Universit\'{e}, CNRS \\ISIR}
  \city{F-75005 Paris}
  \country{France}}
\email{nicolas.bredeche@isir.upmc.fr}

\author{Nicolas Maudet}
\affiliation{
  \institution{Sorbonne Universit\'{e}, CNRS \\ LIP6}
  \city{F-75005 Paris}
  \country{France}}
\email{nicolas.maudet@lip6.fr}

\author{Jason R. Marden}
\affiliation{
  \institution{University of California, Santa Barbara}
  \city{Santa Barbara}
  \country{USA}}
\email{jrmarden@ece.ucsb.edu}

\fi

\begin{abstract}
Conventional distributed approaches to coverage control may suffer from lack of convergence and poor performance, due to the fact that agents have limited information, especially in non-convex discrete environments.
To address this issue, we extend the approach of \cite{marden2016role} 
which demonstrates how a limited degree of inter-agent communication can be exploited to overcome such pitfalls in one-dimensional discrete environments. 
The focus of this paper is on extending such results
to general dimensional settings. 
We show that the extension is convergent and keeps the approximation ratio of 2, meaning that 
any stable solution is guaranteed to have a performance within 50\% of the optimal one.
The experimental results exhibit that our algorithm outperforms several state-of-the-art algorithms, and also that the runtime is scalable.
\end{abstract}


\newcommand{\BibTeX}{\rm B\kern-.05em{\sc i\kern-.025em b}\kern-.08em\TeX}

\usepackage{textcomp}
\usepackage{xcolor}

\usepackage{comment}
\usepackage{url}
\usepackage[T1]{fontenc}
\usepackage[utf8]{inputenc}
\usepackage{blindtext}
\usepackage{pdfpages}

\usepackage{here}
\usepackage{fancyhdr}
\excludecomment{hdn}

\newtheorem{lemma}{Lemma}

\newtheorem{theorem}{Theorem}
\newtheorem{example}{Example}

\usepackage{algorithm}
\usepackage[noend]{algpseudocode}

\algdef{SE}[FOR]{ForTo}{EndFor}[2]{\algorithmicfor\ #1\ \algorithmicto\ #2\ \algorithmicdo}{\algorithmicend\ \algorithmicfor}%
\algtext*{EndFor}

\newcommand{\mc}{\mathcal}

\newcommand{\argmax}{\mathop{\rm argmax~}\limits}
\newcommand{\argmin}{\mathop{\rm argmin~}\limits}

\newcommand{\ball}{{\rm ball}}

\newcommand{\adj}{{\rm adj}}

\newcommand{\lmin}{{\rm lmin}}

\newcommand{\neigh}{{E}}

\newcommand{\newcolor}{black}
\newcommand{\memo}[1]{\textcolor{red}{#1}}
\newcommand{\aamas}[1]{\textcolor{black}{#1}}
\newcommand{\aamasoct}[1]{\textcolor{black}{#1}}

\newcommand{\editNico}[1]{\textcolor{black}{#1}}
\newcommand{\forecai}[1]{\textcolor{black}{#1}}

\def\arxiv{0}
\newcommand{\cameracolor}{black}

\if\arxiv0
\newcommand{\camera}[1]{\textcolor{\cameracolor}{#1}}
\else
\newcommand{\camera}[1]{\textcolor{\cameracolor}{#1}}
\fi

\usepackage{balance} 
\usepackage{tikz}
\usetikzlibrary{arrows,shapes,backgrounds}

\newcommand*\circled[1]{\tikz[baseline=(char.base)]{
            \node[shape=circle,draw,thick,inner sep=2pt] (char) {#1};}}
            
\newcommand*\circledfilled[1]{\tikz[baseline=(char.base)]{
            \node[shape=circle,draw,thick,inner sep=2pt, fill=green!10] (char) {#1};}}
            
\newcommand*\circledred[1]{\tikz[baseline=(char.base)]{
            \node[shape=circle,draw,thick,inner sep=2pt, fill=red!10] (char) {#1};}}
            
\newcommand*\circledblue[1]{\tikz[baseline=(char.base)]{
            \node[shape=circle,draw,thick,inner sep=2pt, fill=blue!10] (char) {#1};}}
            
\newcommand*\circledorange[1]{\tikz[baseline=(char.base)]{
            \node[shape=circle,draw,thick,inner sep=2pt, fill=orange!30] (char) {#1};}}    
            
\newcommand*\circledbrown[1]{\tikz[baseline=(char.base)]{
            \node[shape=circle,draw,thick,inner sep=2pt, fill=brown!30] (char) {#1};}}

\newcommand*\circledfilleds[1]{\tikz[baseline=(char.base)]{
            \node[shape=circle,draw,thick,inner sep=1pt, fill=green!10] (char) {#1};}}
            
\newcommand*\circledreds[1]{\tikz[baseline=(char.base)]{
            \node[shape=circle,draw,thick,inner sep=1pt, fill=red!10] (char) {#1};}}
            
\newcommand*\circledblues[1]{\tikz[baseline=(char.base)]{
            \node[shape=circle,draw,thick,inner sep=1pt, fill=blue!10] (char) {#1};}}
            
\newcommand*\circledoranges[1]{\tikz[baseline=(char.base)]{
            \node[shape=circle,draw,thick,inner sep=1pt, fill=orange!30] (char) {#1};}}

\newcommand*\robot[1]{\tikz[baseline=(char.base)]{
        \node[shape=circle,draw,minimum size=4.5mm, inner sep=0pt, thick, fill=black!20] (char)
        {\rule[-3pt]{0pt}{\dimexpr2ex+2pt}#1};}}
        
\newcommand*\robotred[1]{\tikz[baseline=(char.base)]{
        \node[shape=circle,draw,minimum size=4.5mm, inner sep=0pt, thick, fill=red!50] (char)
        {\rule[-3pt]{0pt}{\dimexpr2ex+2pt}#1};}}
        
\newcommand*\robotgreen[1]{\tikz[baseline=(char.base)]{
        \node[shape=circle,draw,minimum size=4.5mm, inner sep=0pt, thick, fill=green!50] (char)
        {\rule[-3pt]{0pt}{\dimexpr2ex+2pt}#1};}}
        
\newcommand*\robotblue[1]{\tikz[baseline=(char.base)]{
        \node[shape=circle,draw,minimum size=4.5mm, inner sep=0pt, thick, fill=blue!50] (char)
        {\rule[-3pt]{0pt}{\dimexpr2ex+2pt}#1};}}
        
\newcommand*\robotorange[1]{\tikz[baseline=(char.base)]{
        \node[shape=circle,draw,minimum size=4.5mm, inner sep=0pt, thick, fill=orange!80] (char)
        {\rule[-3pt]{0pt}{\dimexpr2ex+2pt}#1};}}
        
\newcommand*\robotbrown[1]{\tikz[baseline=(char.base)]{
        \node[shape=circle,draw,minimum size=4.5mm, inner sep=0pt, thick, fill=brown!70] (char)
        {\rule[-3pt]{0pt}{\dimexpr2ex+2pt}#1};}}
        
\newcommand*\robotviolet[1]{\tikz[baseline=(char.base)]{
        \node[shape=circle,draw,minimum size=4.5mm, inner sep=0pt, thick, fill=violet!50] (char)
        {\rule[-3pt]{0pt}{\dimexpr2ex+2pt}#1};}}

\newcommand*\robotgreens[1]{\tikz[baseline=(char.base)]{
        \node[shape=circle,draw,minimum size=4.mm, inner sep=0pt, thick, fill=green!50] (char)
        {\rule[-3pt]{0pt}{\dimexpr2ex+2pt}#1};}}
        
\newcommand*\robotblues[1]{\tikz[baseline=(char.base)]{
        \node[shape=circle,draw,minimum size=4.mm, inner sep=0pt, thick, fill=blue!50] (char)
        {\rule[-3pt]{0pt}{\dimexpr2ex+2pt}#1};}}
        
\newcommand*\robotoranges[1]{\tikz[baseline=(char.base)]{
        \node[shape=circle,draw,minimum size=4.mm, inner sep=0pt, thick, fill=orange!80] (char)
        {\rule[-3pt]{0pt}{\dimexpr2ex+2pt}#1};}}

\newcommand*\zoneviolet{\tikz[baseline=(char.base)]{
        \node[shape=circle, minimum size=4.5mm, inner sep=0pt, fill=violet!30] (char){} ;}}
        
\newcommand*\zoneorange{\tikz[baseline=(char.base)]{
        \node[shape=circle, minimum size=4.5mm, inner sep=0pt, fill=orange!30] (char){} ;}}
        
\newcommand*\zonered{\tikz[baseline=(char.base)]{
        \node[shape=circle, minimum size=4.5mm, inner sep=0pt, fill=red!10] (char){} ;}}
  
\newcommand*\zoneblue{\tikz[baseline=(char.base)]{
        \node[shape=circle, minimum size=4.5mm, inner sep=0pt, fill=blue!10] (char){} ;}}

\makeatletter
\gdef\@copyrightpermission{
	\begin{minipage}{0.3\columnwidth}
		\href{https://creativecommons.org/licenses/by/4.0/}{\includegraphics[width=0.90\textwidth]{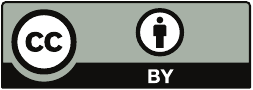}}
        \end{minipage}\hfill
	\begin{minipage}{0.7\columnwidth}
		\href{https://creativecommons.org/licenses/by/4.0/}{This work is licensed under a Creative Commons Attribution International 4.0 License.}
	\end{minipage}
	\vspace{5pt}
}
\makeatother

\begin{document}


\pagestyle{fancy}
\fancyhead{}


\maketitle 


\section{Introduction}

Coverage control is a fundamental problem in the field of multiagent systems (MAS). The objective of a coverage control problem is to deploy homogeneous agents to maximize a given objective function, which basically captures how distant the group of agents as a whole is from a pre-defined set of Points of Interest (PoI). Coverage control has a wide range of applications, such as tracking, mobile sensing networks or formation control of autonomous mobile robots~\cite{cortes2004coverage}. 

It is known that, even in a centralized context, finding an optimal solution for the coverage problem is an NP-hard problem~\cite{megiddo1984complexity}. Hence, most studies focus on approximate approaches. 
In distributed settings, game-theoretical control approaches seek to design agents that will be incentivized to behave autonomously in a way that is well-aligned with the designer's objective. This strategy has proven to be successful in a number of applications (see, e.g.~\cite{douchan2019}). 
The situation is made more difficult in practice since agents often have restricted sensing and communication capabilities. Consequently, agents must make decisions based on local information about their environment and the other agents. 
Unfortunately, algorithms based on local information may \aamas{suffer from lack of convergence and degraded performance}
due to miscoordination between the agents. \aamas{As for the convergence issue, it is known that a move of an agent can affect the cost of agents outside of the neighborhood, and thus, the decrease of the global cost cannot be guaranteed locally, especially in the case of discrete non-convex environment \cite{yun_rus_2014}. 
The degradation of performance can be explained with} a worst-case scenario in which only a single agent can perceive a large number of valuable locations within an environment, while a number of other agents cannot perceive these locations.  
 

Recently, Marden \cite{marden2016role} made more precise the  connection between the degree of locality
related  to the  available  information and the achievable efficiency guarantees. He showed that the achievable approximation ratio depends on the individual amount of information available to the agents. Consequently, distributed algorithms are inevitably subjected to poor worst-case guarantees because of the locality of the information used to make decisions. 
If all agents have full global information as in the case of centralized control, there exist decentralized algorithms that give a 2 approximation ratio. Conversely, under limited information (e.g. Voronoi partitions), no such algorithm provides such an approximation ratio. Rather, the best decentralized algorithm that adheres to these informational dependencies achieves, at most, an $n$ approximation factor, where $n$ is the number of agents.
Then, the focus in MAS settings is on how to design agents that, through sharing a limited amount of information with neighborhood communications, achieve an approximation ratio close to 2.


Indeed, different settings exist  \camera{depending on the assumed communication model, that is, how  information can be shared within the system in order to potentially coordinate moves}:

\begin{enumerate}
    \item agents may not communicate any information and can only be  guided by their local perception; 
        \item  agents may communicate to their  neighbors only (where neighbors can be defined as agents in a limited range, or connected via a network of communication); 
    \item agents may communicate beyond their neighbors, either via broadcasting, or indirectly via gossip protocols.
    
\end{enumerate}

Existing game-theoretical approaches can be classified into these types: classical Voronoi-based best-response approaches fall into either (1) or (2), depending on the assumption. 
Two recent approaches in type (3) 
explore different directions.
Sadeghi \textit{et al.} \cite{sadeghi2020approximation} proposed a distributed algorithm \aamas{for non-convex discrete settings} in which agents have the possibility to coordinate a move with a single other agent (meaning that an agent moves and assumes another agent takes her place simultaneously), possibly beyond their neighbors, when individual moves are not sufficient. \aamas{However, the algorithm sacrifices the approximation ratio for convergence.}
\camera{Another related distributed approach proposed in  \cite{durham2011discrete} is to aim for pairwise optimal partitions (which, they show, are also Voronoi partitions with the additional property that no pairs of agents can cooperate to find a better partitioning of their two regions). Again, while their approach ensures convergence (to good solutions in practice), it does not come with a guaranteed approximation ratio.}
In a similar spirit, Marden~\cite{marden2016role} allows coordinated moves with several \camera{(possibly, beyond two)} agents, but only under the restriction that these agents are within the neighborhood. 
While this approach achieves \aamas{a convergent algorithm with} an approximation ratio of 2,  by sharing only the information of the minimum utility among agents, it is limited in that the only investigated case is the one-dimensional (line) environment. This severely limits real-world applications. 

Other studies \camera{rely on different techniques. Distributed constraint optimization is a natural way to model this problem \cite{zivan-etal-jaamas15}, and several standard algorithms can be exploited. However, they do not offer approximation guarantees either.}
Finally, it is possible to try to achieve global optimality by developing approaches akin to simulated annealing. For example, \cite{arslan2007autonomous} attains global optimum with Spatial Adaptive Play (SAP, a.k.a Boltzmann exploration) and uses random search to escape from the local optimum. However, this approach suffers from a slow convergence rate when the search space is large. There is also no discussion about the relationship between information and efficiency.
Note that we focus on the coverage problem, and  multiagent path-planning issues  \cite{galceran2013survey,stern2019multi,okumura2022solving} are out of the scope of the paper. 

In this paper, we extend the algorithm of \cite{marden2016role} to any-dimensional, \aamas{non-convex discrete} space and compare this approach with the aforementioned alternative variants of game-theoretical control. 
The remainder of this paper is as follows. 
Section~\ref{sec:model} introduces the model and existing approaches. In Section \ref{sec:algo} we detail the algorithm and prove that it 
guarantees \aamas{convergence to} a neighborhood optimum solution with an approximation ratio of 2 without any restriction on the dimensionality of the environment, i.e., the same guarantee as in the 1D case. Additionally, we propose a scalable extension of the algorithm and discuss the computation and communication complexity.
Experiments reported in Section \ref{sec:exps} indeed show that our algorithm outperforms existing ones, along with adhering to the theoretical approximation ratio.
Furthermore, the runtime results confirm the scalability of the proposed algorithm.

\section{Model}
\label{sec:model}


\subsection{Coverage Problems}
\label{sec:ra}

We start with a set of agents $\mc{N}=\{1,\ldots,n\}$ and a set of resources $C=\{c_1,\ldots,c_m\}$.   
In a coverage problem, resources are locations (or points) in a connected metric space. \forecai{We assume that this is discrete finite space that is modelled as a connected graph $(C,\mc{E})$, where $\mc{E}$ is the set of edges that connect two adjacent} \camera{resources}. Though our approach can be extended to continuous settings, we omit the detail due to the space limit.
We denote the distance between two \camera{resources} $a,b\in C$ as $|a-b|$ that is the length of the shortest path. 

An allocation $x$ maps each agent $i$ in  $\mc{N}$ to a resource (i.e., a point) in $C$. An allocation is thus defined as a vector of resources $x=\langle x_i\in C|i\in\mc{N} \rangle$ where $x_i$ is the resource assigned to agent $i$. Note that each agent must be allocated one and only one resource.
An allocation is exclusive i.e., $x_i \neq x_j$ $ \forall i, j \in \mc{N}$ such that $i \neq j$. 
We denote the set of all possible allocations as $\mc{X}=\underset{\substack{i\in\mc{N}}}{\times}\mc{X}_i$ where $\mc{X}_i$ is the set of possible positions for agent $i$. 

Let $g:\mathbb{R}_+\to\mathbb{R}_+$ denote a non-increasing function, $v_c\in\mathbb{R}_+$ be the weight of \camera{resource} $c\in C$, and $\mc{C}\subseteq C$ be a partial space. Then, the objective function for $\mc{C}$  
is defined as follows:
\begin{equation}
\left.
\begin{array}{l}
G(x;\mc{C})=\underset{\substack{c\in\mc{C}}}{\sum}\underset{\substack{i\in\mc{N}}}{\max}~ v_c ~g(|x_i-c|).
\end{array}
\right.
\label{eq:Gdisc}
\end{equation}

For simplicity, we denote $G(x)=G(x,C)$.  The goal of the coverage problem is then to find an optimal allocation $x^*\in\mc{X}$ such that:

\begin{equation}
\left.
\begin{array}{l}
x^*\in\underset{\substack{x\in\mc{X}}}{\argmax}~G(x).
\end{array}
\right.
\label{eq:globalobj}
\end{equation}


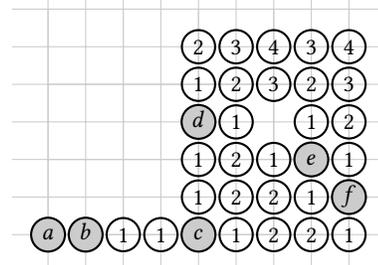
\begin{figure}
\centering
\tikzstyle{background grid}=[draw, black!20,step=.5cm]
\begin{tikzpicture}[show background grid]
\draw[black!20] (0,0) grid (4,3);
            \draw (0,0) node {$\robot{$a$}$}; 
            \draw (0.5,0) node {$\robot{$b$}$};
            \draw (1,0) node {$\circled{1}$};
            \draw (1.5,0) node {$\circled{1}$};
            \draw (2,0) node {$\robot{$c$}$};
            \draw (2.5,0) node {$\circled{1}$};
            \draw (3,0) node {$\circled{2}$};
            \draw (3.5,0) node {$\circled{2}$};
            \draw (4,0) node {$\circled{1}$};
            \draw (2,0.5) node {$\circled{1}$};
            \draw (2.5,0.5) node {$\circled{2}$};
            \draw (3,0.5) node {$\circled{2}$};
            \draw (3.5,0.5) node {$\circled{1}$};
            \draw (4,0.5) node {$\robot{$f$}$};
            \draw (2,1) node {$\circled{1}$};
            \draw (2.5,1) node {$\circled{2}$};
            \draw (3,1) node {$\circled{1}$};
            \draw (3.5,1) node {$\robot{$e$}$};
            \draw (4,1) node {$\circled{1}$};
            \draw (2,1.5) node {$\robot{$d$}$};
            \draw (2.5,1.5) node {$\circled{1}$};
            \draw (3.5,1.5) node {$\circled{1}$};
            \draw (4,1.5) node {$\circled{2}$};
            \draw (2,2) node {$\circled{1}$};
            \draw (2.5,2) node {$\circled{2}$};
            \draw (3,2) node {$\circled{3}$};
            \draw (3.5,2) node {$\circled{2}$};
            \draw (4,2) node {$\circled{3}$};
            \draw (2,2.5) node {$\circled{2}$};
            \draw (2.5,2.5) node {$\circled{3}$};
            \draw (3,2.5) node {$\circled{4}$};
            \draw (3.5,2.5) node {$\circled{3}$};
            \draw (4,2.5) node {$\circled{4}$};
\end{tikzpicture}
\vspace{-0.2cm}
\caption{Coverage example with 6 agents: Circles in grey are agents. A number in a white circle shows the Manhattan distance to the closest agent.} \label{fig:ex1}
\end{figure}

\begin{example}
\label{ex:cov}
Let us then consider the coverage problem depicted in Figure \ref{fig:ex1}. For the sake of exposure, let us assume that 
$g(d)=1/(1+d)$, and $|\cdot|$ is Manhattan distance. Agents are identified by letters $a, b, \dots f$. The environment is a grid world, where \camera{each junction is a resource (i.e. location)}. Circled locations are valued $v_c=1$, while the others are valued $v_c=0$. The locations covered by an agent are represented in grey and we indicate the name of the agent. For unoccupied locations, we indicate on the node the (Manhattan) distance to the closest agent. Thus, a covered location gives a utility of 1, while a location at a distance of 1 from the closest agent gives a utility 1/2, and so on. 
The depicted allocation $x$ has value $G(x) = 1 + 1+ \frac{1}{2} + \dots = 16.4 $. 
It is clearly sub-optimal. The reader can check that the allocation $x'$ where agent $e$ has moved one location up would induce $G(x') \simeq 16.8$. 

\end{example}

\editNico{
It is well-known that solving optimally this problem is NP-hard \cite{megiddo1984complexity}. Nevertheless, thanks to the submodular nature of the objective function $G$, the centralized greedy algorithm which allocates agents one by one (starting from the empty environment) guarantees $1-1/e\approx$ 63\% of optimal \cite{calinescu2011maximizing}. We will use this algorithm as a baseline for comparison. 
}

\subsection{Game-Theoretic Control: Generalities}

Since we focus on distributed algorithms, each agent has to make a choice about her position based on the partial information she has about the other agents. We thus adopt a game theoretic approach where each agent computes her best response based on her individual information \aamasoct{\cite{nisan2007algorithmic,Shoham2008}}. 



In the following,  $-i$ will denote the set of agents excluding $i$: $-i=\mc{N}\setminus\{i\}$. A partial allocation for a subset of agents $S\subseteq \mc{N}$ will be denoted as $x_{S}=\{x_i\}_{i\in S}$.

\aamasoct{This paper will focus on providing each agent $i$ with a utility function of the form $u_i: \mc{X}\to\mathbb{R}$ that will ultimately guide their individual behavior.}
In the rest of the paper, we formulate the set of choices of agent $i$ as $\mc{X}_i(x)$ assuming that it depends on a given allocation of all agents. 
For simplicity, we denote $\mc{X}(x)=\underset{\substack{i\in\mc{N}}}{\times}\mc{X}_i(x)$.

When each agent $i\in\mc{N}$ selects the position  $x_i\in\mc{X}_i(x)$ that maximizes her utility given the other agent's positions $x_{-i}$, the resulting allocation $x = \langle x_1 \cdots x_n \rangle $ can be a (pure) Nash equilibrium such that:

\begin{equation}
\left.
\begin{array}{l}
\forall i \in \mc{N}  \textrm{~}  u_i((x_i,x_{-i}))\geq u_i((x_i',x_{-i})), \forall x_i'\in \mc{X}_i(x), x_i \neq x_i'.
\end{array}
\right.
\label{eq:ne}
\end{equation}

In general, the efficiency of a Nash equilibrium $G(x)$ can be smaller than the optimal value $G(x^*)$. 
\aamasoct{One of the reasons for this suboptimality is the miscoordination among self-interested agents, as agents are required to make independent decisions in response to available information.  Other sources of suboptimality come from the structure of the agents' utility functions and available choice sets.}
The (worst-case, among all instances) ratio between the worst Nash equilibria and the social optimum is known as the price of anarchy (PoA) \cite{KoutsoupiasEtAl2009}. 
Here, the choice set $\mc{X}_i$ could encode physical constraints on choices, i.e. an agent can only physically choose \camera{among a} subset of choices given the behavior of the collective $x$. Alternatively, the choice set could encode information availability of the agent, e.g. it is the set of choices for which the agent can evaluate its utility. Regardless of the interpretation, it is important to highlight that the structure of the choice sets can significantly alter the structure and efficiency of the resulting equilibria which we discuss in the next section.

\subsection{Local Information}
\label{sec:info}


Several approaches in game theory seek to exploit the fact that agents typically have limited sensing power and only a local view of the situation. 
Marden \cite{marden2016role} \aamasoct{analyzed how miscoordination among agents with limited information leads to inefficient Nash equilibria. To this end, he introduced the concept of \emph{information set} which is the set of choices each agent can perceive and compute the resulting utilities.
In this paper, $\mc{X}_i(x)$ corresponds to the information set that is the set of locations agent $i$ can perceive based on spatial proximity. With this notion, an allocation is a Nash equilibrium (Equation ~\ref{eq:ne}) if agents are at least as happy as any choice for which they can evaluate their utility.}
Observe that the information set $\mc{X}_i(x)$ is a state-dependent notion: the local information available may vary depending on the current allocation $x$. 

To model to what extent the information is localized, Marden defined
the following \emph{redundancy index} associated to the agents' local information sets $\{\mc{X}_i\}_{i \in \mc{N}}$:

\begin{equation}
\left.
\begin{array}{l}
f=\underset{\substack{x\in\mc{X}}}{\min}\underset{\substack{y\in C}}{\textrm{~}\min} \textrm{~} |\{i\in\mc{N}:y\in\mc{X}_i(x)\}|.
\end{array}
\right.
\label{eq:f}
\end{equation}

Intuitively, $f$ represents the minimum number of agents that perceive the same resource available for their choice. 
In particular, we note that: 
\begin{itemize}
    \item $f>0$ guarantees that all the locations are always  a possible choice for some agent of the system; 
    \item If there is an allocation  $x$ where  a resource is a possible choice for only one agent and all other resources are possible choices for at least one agent,  then $f=1$. 
    \item $f = n$ is the extreme case of full information access where all resources are possible choices for all agents. 
\end{itemize}

\subsection{Linking Information and Inefficiency}

Intuitively, local information can cause distributed systems based on game-theoretic control to get stuck in an inefficient allocation. Indeed, some agents who may obtain higher utilities for a resource may not have  access to this information.  The redundancy index hence gives insights into the amount of information available to the agents and about guarantees on the worst-case ratio. 
  
Marden \cite{marden2016role} investigated this interplay, in the broader context of resource allocation games. 
The global objective function $G$ \camera{is assumed to be} monotone submodular i.e., \camera{it }satisfies the following conditions:
\begin{equation}
\left.
\begin{array}{l}
G(x_{T})\geq G(x_{S}), \forall S\subseteq T\subseteq N. \\
G(x_S)-G(x_{S\setminus \{i\}})\geq G(x_T)-G(x_{T\setminus\{i\}}),  \forall S\subseteq T \subseteq \mc{N}, \forall i\in S.\\
\end{array}
\right.
\label{eq:submodular}
\end{equation}
and satisfies two further properties. First, the utility of each agent is greater than her marginal contribution:
\begin{equation}
\left.
\begin{array}{l}
u_i(x) \geq G(x)-G(x_{-i}).
\end{array}
\right.
\label{eq:wlu}
\end{equation}
Second, social welfare is less than the global objective:
\begin{equation}
\left.
\begin{array}{l}
\underset{\substack{i\in\mc{N}}}{\sum}u_i(x) \leq G(x).
\end{array}
\right.
\label{eq:bb}
\end{equation}
In the later section, we will see that the global objective function and the utility function of the conventional coverage control  both satisfy these assumptions.


The following theorem shows how the value of  $f$ impacts the efficiency of  Nash equilibrium allocations:

\begin{theorem}[from \cite{marden2016role}]
If $G$ is a monotone submodular set function and $u_i$ satisfies (\ref{eq:wlu}) and (\ref{eq:bb}), the worst case efficiency of Nash equilibrium $x$ is lower bounded by
\begin{equation}
\left.
\begin{array}{l}
G(x)\geq \frac{f}{n+f}G(x^*).
\end{array}
\right.
\label{eq:poa}
\end{equation}
Furthermore, there exists a case such that: 
\begin{equation}
\left.
\begin{array}{l}
G(x)= \frac{f}{n}G(x^*).
\end{array}
\right.
\label{eq:poaeq}
\end{equation}
\label{thm:poa}
\end{theorem}
Thus, we know that the approximation ratio of a distributed allocation algorithm can be 2 in the case of full information ($f=n$), because $G(x)\geq \frac{1}{2}G(x^*)$  in this case (from Equation \ref{eq:poa}). Also, it is impossible to guarantee an approximation ratio better than $n$ 
in case of $f=1$ (from Equation \ref{eq:poaeq}). We will use this result in the later section to analyze the efficiency of coverage control algorithms.

\subsection{Voronoi-Based Control}
\label{sec:distcontrol}

\camera{Most} standard approaches for coverage control \camera{are} based on Voronoi partitioning \cite{gao2008notes}. 
A Voronoi partition divides the space into local regions for each agent. Formally, for a given allocation $x$, a Voronoi partition $\mc{V}_i(x;\mc{C})$ of space $\mc{C}$ is defined as follows:
\begin{equation}
\left.
\begin{array}{l}
\mc{V}_i(x;\mc{C})=\{c\in\mc{C}|i=\underset{\substack{j\in\mc{N}}}{\argmin}|x_j-c|\}.

\end{array}
\right.
\label{eq:vor}
\end{equation}
When ties occur (i.e., when several agents are at the same minimal distance from some location), agents are prioritized lexicographically. Let $\mc{P}=(\mc{P}_1,\ldots,\mc{P}_N)$ denote a partition of the space. In the rest of the paper, we define the utility function of agent $i$ depending on partition $\mc{P}$ as follows:
\begin{equation}
\left.
\begin{array}{l}
u_i(x,\mc{P})=\underset{\substack{c\in\mc{P}_i}}{\sum}v_c ~ g(|x_i-c|).
\end{array}
\right.
\label{eq:udisccov}
\end{equation}

In this subsection, we assume that the locations an agent can perceive are limited to those within their Voronoi region, that is $\mc{P}_i=\mc{X}_i(x)=\mc{V}_i(x;C)$. Therefore, the utility function (Equation \ref{eq:udisccov}) satisfies the assumptions of Equation \ref{eq:wlu} and Equation \ref{eq:bb}. 
We also denote the neighborhood of $i$ as the agents connected in the dual Delaunay graph of the Voronoi partition of $i$, i.e., the neighborhood of $\mc{N}_i\subseteq \mc{N}$ are the agents $j \in \mc{N}$ whose  Voronoi region is connected to the Voronoi region of $i$. 
Voronoi partition plays an important role in distributed coverage control because 
agents can compute the best responses improving the objective function locally, with limited communication. 
The process is just a sequence of best response updates (in the sense defined above) of the different agents to compute their next locations in their local Voronoi region. Once agents are assigned these new locations, Voronoi regions are updated and the process iterates, until convergence. \aamas{However, in the non-convex discrete} setting the move of an agent within its partition can affect not only neighbors, but \aamas{the whole set of agents in the worst case \cite{yun_rus_2014}. We will see an example later in Figure \ref{fig:externality}.}  
In theory, the guarantee of convergence requires avoiding moves that could impact beyond an agent's neighbors \cite{yun_rus_2014, sadeghi2020approximation}.

Observe that the Voronoi partition induces a redundancy index of $f=1$ because all the points in Voronoi regions $\mc{X}_i$ are available only for agent $i$. Then, the efficiency $G(x)$ can be $1/n$ of the optimum value in the worst case. 

\begin{figure}
\begin{minipage}{0.25\textwidth}
\centering
\tikzstyle{background grid}=[draw, black!20,step=.5cm]
\begin{tikzpicture}[show background grid]
\draw[black!20] (0,0) grid (4,2.5);
            \draw (0,0) node {$\robotviolet{$a$}$}; 
            \draw (0.5,0) node {$\robotorange{$b$}$};
            \draw (1,0) node {$\circledorange{1}$};
            \draw (1.5,0) node {$\circledblue{1}$};
            \draw (2,0) node {$\robotblue{$c$}$};
            \draw (2.5,0) node {$\circledblue{1}$};
            \draw (3,0) node {$\circledblue{2}$};
            \draw (3.5,0) node {$\circledfilled{2}$};
            \draw (4,0) node {$\circledbrown{1}$};
            \draw (2,0.5) node {$\circledblue{1}$};
            \draw (2.5,0.5) node {$\circledblue{2}$};
            \draw (3,0.5) node {$\circledfilled{2}$};
            \draw (3.5,0.5) node {$\circledfilled{1}$};
            \draw (4,0.5) node {$\robotbrown{$f$}$};
            \draw (2,1) node {$\circledred{1}$};
            \draw (2.5,1) node {$\circledred{2}$};
            \draw (3,1) node {$\circledfilled{1}$};
            \draw (3.5,1) node {$\robotgreen{$e$}$};
            \draw (4,1) node {$\circledfilled{1}$};
            \draw (2,1.5) node {$\robotred{$d$}$};
            \draw (2.5,1.5) node {$\circledred{1}$};
            \draw (3.5,1.5) node {$\circledfilled{1}$};
            \draw (4,1.5) node {$\circledfilled{2}$};
            \draw (2,2) node {$\circledred{1}$};
            \draw (2.5,2) node {$\circledred{2}$};
            \draw (3,2) node {$\circledred{3}$};
            \draw (3.5,2) node {$\circledfilled{2}$};
            \draw (4,2) node {$\circledfilled{3}$};
            \draw (2,2.5) node {$\circledred{2}$};
            \draw (2.5,2.5) node {$\circledred{3}$};
            \draw (3,2.5) node {$\circledred{4}$};
            \draw (3.5,2.5) node {$\circledfilled{3}$};
            \draw (4,2.5) node {$\circledfilled{4}$};
            \foreach \x in {0.5,1,1.5,2,2.5}
            \draw (0,\x) node {$\zoneviolet$};
            \foreach \x in {0.5,1,1.5,2,2.5}
            \draw (0.5,\x) node {$\zoneorange$};
            \foreach \x in {0.5,1}
            \draw (1,\x) node {$\zoneorange$};
            \draw (1.5,0.5) node {$\zoneblue$};
            \foreach \x in {1,1.5,2,2.5}
            \draw (1.5,\x) node {$\zonered$};
            \foreach \x in {1.5,2,2.5}
            \draw (1,\x) node {$\zonered$};
            \draw (3,1.5) node {$\zonered$};
\end{tikzpicture}
\end{minipage}\hfill
    \begin{minipage}{0.2\textwidth}
\begin{tikzpicture}[scale=0.9]
\draw[black!00] (0,0) grid (3.5,2.5);
            \draw (0,0) node (roba) {$\robotviolet{$a$}$};
            \draw (1,0) node (robb) {$\robotorange{$b$}$};
            \draw (2,0) node (robc) {$\robotblue{$c$}$};
            \draw (2,1.5) node (robd) {$\robotred{$d$}$};
            \draw (3,1) node (robe) {$\robotgreen{$e$}$};
            \draw (3.5,0) node (robf) {$\robotbrown{$f$}$};
            \draw (roba) -- (robb); 
            \draw (robb) -- (robc);
            \draw (robc) -- (robd);
            \draw (robe) -- (robd);
            \draw (robc) -- (robe);
            \draw (robb) -- (robd);
            \draw (robe) -- (robf);
\end{tikzpicture}

    \end{minipage}
    \caption{The environment of Example \ref{ex:cov} with Voronoi partitions and the corresponding neighborhood graph}
\label{fig:neighbors}
\end{figure}
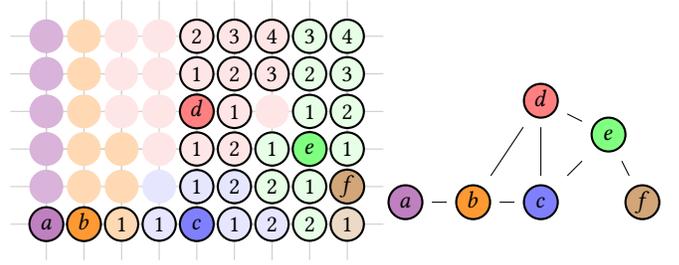

\begin{example} (Ex. 1, cont.). 
Figure~\ref{fig:neighbors} indicates the Voronoi region of each agent by coloring the locations in the same color. Recall that ties are broken lexicographically. Figure~\ref{fig:neighbors} gives the corresponding neighboring graph: agent $e$ has 3 neighbors $\{c,d,f\}$.  
The utilities of the agents are as follows: $u_a= 1$, $u_b = 1.5$, 
$u_c \simeq 3.2$ 
$u_d = 4.2$ 
$u_e \simeq 5 $ 
and $u_f= 1.5$ 

\end{example}
\medskip

\camera{As mentioned in the introduction, the communication requirements of these approaches is very low, as agents never need to exchange beyond  their neighbors in the Delaunay graph (note that this can also be achieved with a range-limited gossip protocol, as long as an appropriate motion protocol is designed \cite{durham2011discrete}). 
In the next section, we introduce our protocol which assumes agents can always communicate with their neighbors, and requires them in addition to build a communication tree to spread a limited amount of information to facilitate coordination. Going beyond this communication model is left for future work.}

\section{A Neighborhood Optimal Algorithm}
\label{sec:algo}

To address the inefficiency issue of coverage control through inter-agent communication, we extend the solution for 1-dimensional space \cite{marden2016role} to general settings by making agents share additional information. 


\subsection{Key Notions}
\label{sec:mardenext}


The idea of the algorithm is to incrementally compute an allocation $x$ based on a partition $\mc{P}$, which is not necessarily a Voronoi partition. Note that the neighborhood $\mc{N}_i(\mc{P})$ is also defined over $\mc{P}$, as the set of agents that have partitions adjacent to $\mc{P}_i$.
The algorithm is anytime and updates a solution $(x,\mc{P})$ for each iteration. 




We follow  \cite{marden2016role} and define the maximum gain in $G$ when adding $k$ new agents into space $\mc{C}$, as:

\begin{equation}
\left.
\begin{array}{l}
M_k(x,\mc{C})=\underset{\substack{y_1,\ldots,y_k\in\mc{C}}}{\max}G(y_1,\ldots,y_k,x;\mc{C})-G(x;\mc{C}),
\end{array}
\right.
\label{eq:Mk}
\end{equation}
where this best allocation of the $k$ new agents is denoted as follows:
\begin{equation}
\left.
\begin{array}{l}
B_k(x,\mc{C})=\underset{\substack{y_1,\ldots,y_k\in\mc{C}}}{\argmax}G(y_1,\ldots,y_k,x;\mc{C})-G(x;\mc{C}).
\end{array}
\right.
\label{eq:Bk}
\end{equation}

Note that determining the locations $y_1, \dots y_k$ maximizing Equation~\ref{eq:Bk} is actually an NP-hard problem since it consists in general in  solving a multiagent coverage problem. 

\camera{We denote $\mc{P}_{ij}= \mc{P}_i \cup \mc{P}_j$. 
Informally, we say that an allocation is \emph{neighborhood optimal} when no coalition of agents within the same neighborhood could reallocate themselves in a way which would improve the global objective, and would not benefit enough from accommodating a third agent in their combined regions. 
We will make these notions precise below.  
}
\camera{
But before going into the details, we state our main result:
\begin{theorem}[Convergence with Performance guarantee]
Under our communication model, Algorithm \ref{alg:extmarden} terminates in a neighborhood optimum allocation. Its approximation ratio is 2.
\label{thm:main}
\end{theorem}
Note that this approximation ratio of 2 is equal to the lower bound predicted by Equation~\ref{eq:poa} 
in case of full information ($f=n$). 
In our case it is achieved by only requiring agents to communicate beyond their neighborhood (i) the utility and location of the worst-off agent, and (ii) the identity of the agent which would contribute the most to the global objective by adding an agent in her region. 
In case of ties a deterministic choice mechanism is used, for instance based on agents' id.  
}
 \camera{
Hence, the following notations will be useful: }

\begin{equation*}
\left.
\begin{array}{ll}
i_{\max}^+(x,\mc{P})=\underset{\substack{i\in\mc{N}}}{\argmax}~\!\!M_1(x_i;\mc{P}_i)
&
V(x,\mc{P})=\underset{\substack{i\in\mc{N}}}{\max}~M_1(x_i;\mc{P}_i) \\
i_{\min}(x,\mc{P})=\underset{\substack{i\in\mc{N}}}{\argmin}~\!\!u_i(x;\mc{P})
&
 u_{\min}(x,\mc{P})=\underset{\substack{i\in\mc{N}}}{\min}~u_i(x;\mc{P})
 \end{array}
 \right.
\end{equation*}

\medskip
\begin{example}(Ex. 1, cont.). 
We have $i_{\min} = a$, since $u_a = 1$. 
Furthermore, \camera{the agents which would contribute to $G(x)$ the most from adding a single agent within their regions are agents $e$ and $d$. For instance,}  adding an agent (depicted as '+') would induce $G$ of $6.5$ in her region, thus $M_1(x_e,\mc{P}_e) \simeq 6.5 - 5 = 1.5$  
(See Figure~\ref{fig:addone} for the illustration). \camera{We assume $i^+_{\max} = d$ by lexicographic tie-breaking.}

\end{example}
\medskip

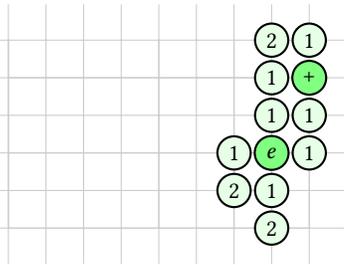
\begin{figure}[h]
\centering
\tikzstyle{background grid}=[draw, black!20,step=.5cm]
\begin{tikzpicture}[show background grid]
\draw[black!20] (0,0) grid (4,2.5);
            \draw (3.5,0) node {$\circledfilled{2}$};
            \draw (3,0.5) node {$\circledfilled{2}$};
            \draw (3.5,0.5) node {$\circledfilled{1}$};
            \draw (3,1) node {$\circledfilled{1}$};
            \draw (3.5,1) node {$\robotgreen{$e$}$};
            \draw (4,1) node {$\circledfilled{1}$};
            \draw (3.5,1.5) node {$\circledfilled{1}$};
            \draw (4,1.5) node {$\circledfilled{1}$};
            \draw (3.5,2) node {$\circledfilled{1}$};
            \draw (4,2) node {$\robotgreen{+}$};
            \draw (3.5,2.5) node {$\circledfilled{2}$};
            \draw (4,2.5) node {$\circledfilled{1}$};
\end{tikzpicture}
\caption{Agent $e$ would contribute to $G(x)$ the most from the addition of an agent (depicted as `+') in her region}
\label{fig:addone}
\end{figure}


\subsection{High Level Description of the Algorithm}

\camera{
The principle of our algorithm follows \cite{marden2016role}, but is adapted to the 2D setting. Essentially, the idea is for neighboring agents to compare what they would gain by optimizing over their combined regions, or by optimizing over their combined regions with a third agent (which is necessarily at least as good). When the difference between these two situations is significant enough (larger than the utility of the worst-off agent), the decision is to make room for a third agent. This way, the space left open can eventually be filled.   
}

\camera{
As mentioned before, agents will communicate a limited amount of information.  
Spreading the information of $i_{\min}$ and $i^+_{\max}$ within the system can be achieved by standard distributed algorithms, eg. flooding and gossip protocols \cite{wattenhoffer2016}. At the end of this process, a communication spanning tree is built. 
Figure \ref{fig:commtree} shows a communication tree obtained for Example \ref{ex:cov}. In our algorithm, interaction takes place in the neighborhood defined by such communication trees, which are updated when required. In the following we denote by \emph{parent(i,$\mc{P}$)} the parent of agent $i$, by \emph{neigh}$(S,\mc{P})$ the neighbors of a set of agents $S$, and by $\neigh(\mc{P})$ the set of all pairs of neighboring agents. All these notions are understood as  restricted to the current communication tree (in its undirected version for the notion of neighborhood). }

As in \cite{marden2016role}, it is useful to classify the solution $(x,\mc{P})$ into 4 states. First, the solution space is divided into the following two states $Z_1$ and $Z_2$, by checking if $i_{\max}^+$ would gain more from adding an agent in her region: 
\begin{align}
Z_1=&\{(x,\mc{P})|V(x,\mc{P})>u_{\min}(x,\mc{P})\},
\label{eq:Z1} \\
Z_2=&\{(x,\mc{P})|V(x,\mc{P})\leq u_{\min}(x,\mc{P})\}. \label{eq:Z2}
\end{align}

In $Z_2$, no single agent would contribute enough from accommodating $i_{\min}$. 
Next the solutions in $Z_2$ are further classified depending on whether integrating a further agent in the neighborhood of two agents could induce a significant enough marginal gain. 
In the following state $Z_3$, it is not the case:
\begin{align}
Z_3=&\{(x,\mc{P})\subseteq Z_2| \nonumber \\ 
&M_{3}(\emptyset,\mc{P}_{ij})-M_{2}(\emptyset,\mc{P}_{ij})\leq u_{\min}(x,\mc{P}), \label{eq:Z3} 
~\forall (i,j) \in \neigh(\mc{P}) \} 
\end{align}
Equation~\ref{eq:Z3} means that the gain in the neighborhood optimum by accommodating $i_{\min}$ in $\mc{P}_{ij}$ cannot be larger than $u_{\min}$. 

Lastly, a solution is classified as state $Z_4$ if  neighborhood optimality \cite{weiss2001optimality} is achieved for all agents. $Z_4$ is the terminal state that is reached  from a solution in $Z_3$ if it satisfies the following condition:

\begin{align}
Z_4=&\{(x,\mc{P})\subseteq Z_3| \nonumber \\
& u_i(x,\mc{P}) + u_j(x,\mc{P})=M_{2}(\emptyset,\mc{P}_{ij}), ~\forall (i,j) \in\neigh(\mc{P})\} \label{eq:Z4}
\end{align}

\medskip
\begin{example}(Ex. 1, cont.)
The allocation $x$ depicted is in state $Z_1$ since $V(x,\mc{P})= 1.5 > 1 = u_a$, with $i_{\min}=a$.

\end{example}
\medskip

\begin{figure}
\begin{tabular}{ll}
\begin{minipage}{0.22\textwidth}
\centering
\tikzstyle{background grid}=[draw, black!20,step=.5cm]
\begin{tikzpicture}[scale=0.9]
\draw[black!00] (0,0) grid (3.5,3);
            \draw (0,0.5) node (roba) {$\robotviolet{$a$}$};
            \draw (0,0) node (ua) {$1$};
            \draw (1,0.5) node (robb) {$\robotorange{$b$}$};
            \draw (1,0) node (ub) {$1.5$};
            \draw (2,0.5) node (robc) {$\robotblue{$c$}$};
            \draw (2,0) node (uc) {$3.2$};
            \draw (2,2) node (robd) {$\robotred{$d$}$};
            \draw (2,2.5) node (ud) {$4.2$};
            \draw (3,1.5) node (robe) {$\robotgreen{$e$}$};
            \draw (3,2) node (ue) {$5$};
            \draw (3.5,0.5) node (robf) {$\robotbrown{$f$}$};
            \draw (3.5,0) node (uf) {$1.5$};
            \draw (roba) -- (robb); 
            \draw (robb) -- (robc);
            \draw (robc) -- (robd);
            \draw (robe) -- (robd);
            \draw (robc) -- (robe);
            \draw (robb) -- (robd);
            \draw (robe) -- (robf);
\end{tikzpicture}
\end{minipage}\hfill &
\begin{minipage}{0.22\textwidth}
\centering
\tikzstyle{background grid}=[draw, black!20,step=.5cm]
\begin{tikzpicture}[scale=0.9]
\draw[black!00] (0,0) grid (3.5,3);
            \draw (0,0.5) node (roba) {$\robotviolet{$a$}$};
            \draw (1,0.5) node (robb) {$\robotorange{$b$}$};
            \draw (2,0.5) node (robc) {$\robotblue{$c$}$};
            \draw (2,2) node (robd) {$\robotred{$d$}$};
            \draw (3,1.5) node (robe) {$\robotgreen{$e$}$};
            \draw (3.5,0.5) node (robf) {$\robotbrown{$f$}$};
            \draw[->] (roba) -- (robb); 
            \draw[->] (robb) -- (robc);
            \draw[->] (robc) -- (robd);
            \draw[->] (robc) -- (robe);
            \draw[->] (robe) -- (robf);
\end{tikzpicture}
\end{minipage}
\end{tabular}
    \caption{(Left): the neighborhood graph in Example \ref{ex:cov}. The numbers represent utilities. 
    (Right): a communication tree in convergence. In this case, agent `a' is $i_{\min}$ and therefore the root of the tree. }
\label{fig:commtree}
\end{figure}
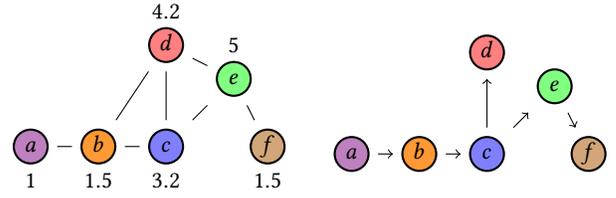

\begin{algorithm}[t]                 
	\caption{Neighborhood optimum algorithm}
	\label{alg:extmarden}
	\begin{algorithmic}[1]
	\Procedure{$x=$NeighborOpt}{$x,\mc{C}$}
		\State{$\mc{P}_i=\mc{V}_i(x;\mc{C}),~\forall i\in\mc{N}$ \label{line:emP0}}
		\While{$(x,\mc{P})\not\in Z_4$ \label{line:emiter}}
        \State{Build communication tree}
		\State{Communicate $u_{\min}(x,\mc{P})$, $x_{i_{\min}}(x,\mc{P})$ and $i_{\max}^+(x,\mc{P})$ \label{line:ems0} \label{line:emconsensus}}
    	\If{$(x,\mc{P})\in Z_1$} 
        \State{$i \leftarrow i_{\max}^+$}
        \Else
    	\State{pick $i \in \mc{N}$ \label{line:empickup}}
        \EndIf
		\State{$j \leftarrow parent(i,\mc{P})$ \label{line:emNi}}
    	\If{$i_{\min}(x,\mc{P}) = j$ \textbf{or} \\ \hspace*{4em}$M_{3}(\emptyset,\mc{P}_{ij})-M_{2}(\emptyset,\mc{P}_{ij})\leq u_{\min}(x,\mc{P})$\label{line:emif}}
    	\State{ $x_{\{i,j\}} \leftarrow B_{2}(\emptyset,\mc{P}_{ij})$  \hfill {\bf (Step a)}\label{line:emopt}}
    	\State{$\mc{P}_k \leftarrow \mc{V}_k(x',\mc{P}_{ij}),~\forall k\in\{i,j\}$}
    	\Else
    	\State{Consider a virtual agent $l$ \hfill {\bf (Step b)}  \label{line:em2bs}}
    	\State{$x' \leftarrow B_{3}(\emptyset,\mc{P}_{ij})$ \label{line:emroom}}
    	\State{$\mc{P}'_{k} \leftarrow \mc{V}_{k}(x',\mc{P}_{ij}),~\forall k\in \mc{N}'_i=\{1,2,3\}$ \label{line:emnewP}} 
     	\State{$\tilde{i}_{\min}= \underset{\substack{k \in neigh(\{i,j\},\mc{P})}}{\argmin} |x_j-x_{i_{\min}}|$ \label{line:emproxy}}  
        \State{$k_l = \underset{\substack{k\in neigh(\tilde{i}_{\min}},\mc{P})\cap\mc{N}'_i} {\argmin} |x'_{k}-x_{\tilde{i}_{\min}}|, ~x_l \leftarrow x'_{k_l}$ \label{line:eml}}
    	\State{$x_{\{i,j\}} \leftarrow x'\setminus \{x_l\}$ \label{line:emmap}}
        \State{$i_+ \leftarrow \underset{\substack{k \in neigh(l,\mc{P})\cap\{i,j\}}}{\argmin}~|x_j-x_{l}|,~\mc{P}_{i_+} \leftarrow \mc{P}_{i_+}\cup\mc{P}_l$ \label{line:emnear}} 
    	\EndIf
		\EndWhile
		\State{Return $(x,\mc{P})$}
	\EndProcedure
	\end{algorithmic}
\end{algorithm}

Given this classification, we propose a distributed algorithm returning an equilibrium solution $(x,\mc{P})$. 
The algorithm is summarized in Algorithm \ref{alg:extmarden}. 
\color{\newcolor}
Firstly, all agents initialize $\mc{P}_i$ with a Voronoi partition \camera{based on geodesic distances \cite{haumann2011discoverage}} (Line \ref{line:emP0}). As far as the state is not $Z_4$, 
agents build a communication tree and share the necessary information, $u_{\min}$, $x_{i_{\min}}$ and $i_{\max}^+$ based on a consensus algorithm via neighborhood communication. 
(Lines \ref{line:emiter}-\ref{line:emconsensus}).

\camera{If the algorithm is in state $Z_1$, the agent $i^+_{\max}$ is picked (Line \ref{line:empickup}), otherwise a random agent is picked. } 
\camera{Note that in a distributed fashion, this could be done for instance by sharing `done/undone' status and agent IDs, to pick up the `undone' agent with the smallest ID. 
Alternatively, agents could probabilistically move (probability $\alpha$) or not (probability $1-\alpha$).  
Agent $i$ retrieves her parent (denoted $j$) in the current communication tree. 
The activated agent $i$ then checks whether the combined regions could  accommodate another agent (Line \ref{line:emif}). }

\begin{itemize}
    \item \textbf{Step a:} If $\mc{P}_{ij}$ cannot accommodate another agent \camera{or $i_{\min}$ belongs to this neighborhood,} agent $i$ computes a neighborhood optimum for the pair of agents $(i,j)$ and the allocation is implemented  (Line \ref{line:emopt}). 
    \item \textbf{Step b:} \camera{If $\mc{P}_{ij}$ can accomodate another agent, 
    it computes new locations $x'$ for agents $i$ and $j$, together with an additional agent $l$ (Line \ref{line:emroom}). 
Then, the algorithm computes a new partition for these $3$ agents by splitting $\mc{P}_{ij}$ based on the new locations $x'$ (Line \ref{line:emnewP}).
To allocate the new partition, the algorithm finds the agent $\tilde{i}_{\min}\not \in \{i,j\}$ that is the closest to $i_{\min}$ in the communication tree. 
Then the partition closest to $\tilde{i}_{\min}$ and adjacent to $\mc{P}_{\tilde{i}_{\min}}$ is allocated to $l$ first (Line \ref{line:eml}).
The other partitions are allocated to $i$ and $j$, for example by an optimal matching algorithm so as to minimize the moves of the agents (Line \ref{line:emmap} \cite{crouse2016implementing}). Lastly, $\mc{P}_l$ is merged to the partition of the agent $i_+$ that is the closest and adjacent to $l$ (Line \ref{line:emnear}).}
\end{itemize}

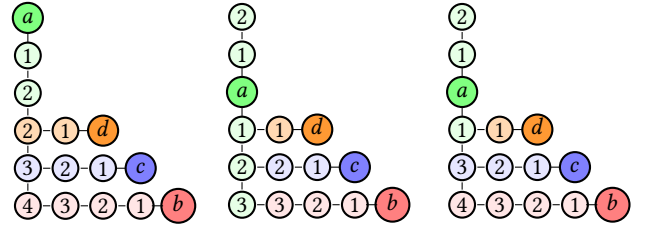
\begin{figure}
\begin{tabular}{ccc}
\begin{minipage}{0.16\textwidth}
\centering
\begin{tikzpicture}
    \draw (0,2.5) node (05) {$\robotgreens{$a$}$};
    \draw (0,2) node (04) {$\circledfilleds{1}$};
    \draw (0,1.5) node (03) {$\circledfilleds{2}$};
    \draw (0,1) node (02) {$\circledoranges{2}$};
    \draw (0.5,1) node (12) {$\circledoranges{1}$};
    \draw (1,1) node (22) {$\robotoranges{$d$}$};
    \draw (0,0.5) node (01) {$\circledblues{3}$};
    \draw (0.5,0.5) node (11) {$\circledblues{2}$};
    \draw (1,0.5) node (21) {$\circledblues{1}$};
    \draw (1.5,0.5) node (31) {$\robotblues{$c$}$};
    \draw (0,0) node (00) {$\circledreds{4}$};
    \draw (0.5,0) node (10) {$\circledreds{3}$};
    \draw (1,0) node (20) {$\circledreds{2}$};
    \draw (1.5,0) node (30) {$\circledreds{1}$};
    \draw (2,0) node (40) {$\robotred{$b$}$};

    \draw (00) -- (01);
    \draw (01) -- (02);
    \draw (02) -- (03);
    \draw (03) -- (04);
    \draw (04) -- (05);
    \draw (02) -- (12);
    \draw (12) -- (22);
    \draw (01) -- (11);
    \draw (11) -- (21);
    \draw (21) -- (31);
    \draw (00) -- (10);
    \draw (10) -- (20);
    \draw (20) -- (30);
    \draw (30) -- (40);
\end{tikzpicture}
\end{minipage}\hfill

\begin{minipage}{0.16\textwidth}
\centering
\begin{tikzpicture}
    \draw (0,2.5) node (05) {$\circledfilleds{2}$};
    \draw (0,2) node (04) {$\circledfilleds{1}$};
    \draw (0,1.5) node (03) {$\robotgreens{$a$}$};
    \draw (0,1) node (02) {$\circledfilleds{1}$};
    \draw (0.5,1) node (12) {$\circledoranges{1}$};
    \draw (1,1) node (22) {$\robotoranges{$d$}$};
    \draw (0,0.5) node (01) {$\circledfilleds{2}$};
    \draw (0.5,0.5) node (11) {$\circledblues{2}$};
    \draw (1,0.5) node (21) {$\circledblues{1}$};
    \draw (1.5,0.5) node (31) {$\robotblues{$c$}$};
    \draw (0,0) node (00) {$\circledfilleds{3}$};
    \draw (0.5,0) node (10) {$\circledreds{3}$};
    \draw (1,0) node (20) {$\circledreds{2}$};
    \draw (1.5,0) node (30) {$\circledreds{1}$};
    \draw (2,0) node (40) {$\robotred{$b$}$};

    \draw (00) -- (01);
    \draw (01) -- (02);
    \draw (02) -- (03);
    \draw (03) -- (04);
    \draw (04) -- (05);
    \draw (02) -- (12);
    \draw (12) -- (22);
    \draw (01) -- (11);
    \draw (11) -- (21);
    \draw (21) -- (31);
    \draw (00) -- (10);
    \draw (10) -- (20);
    \draw (20) -- (30);
    \draw (30) -- (40);
\end{tikzpicture}
\end{minipage}

\begin{minipage}{0.16\textwidth}
\centering
\begin{tikzpicture}
    \draw (0,2.5) node (05) {$\circledfilleds{2}$};
    \draw (0,2) node (04) {$\circledfilleds{1}$};
    \draw (0,1.5) node (03) {$\robotgreens{$a$}$};
    \draw (0,1) node (02) {$\circledfilleds{1}$};
    \draw (0.5,1) node (12) {$\circledoranges{1}$};
    \draw (1,1) node (22) {$\robotoranges{$d$}$};
    \draw (0,0.5) node (01) {$\circledblues{3}$};
    \draw (0.5,0.5) node (11) {$\circledblues{2}$};
    \draw (1,0.5) node (21) {$\circledblues{1}$};
    \draw (1.5,0.5) node (31) {$\robotblues{$c$}$};
    \draw (0,0) node (00) {$\circledreds{4}$};
    \draw (0.5,0) node (10) {$\circledreds{3}$};
    \draw (1,0) node (20) {$\circledreds{2}$};
    \draw (1.5,0) node (30) {$\circledreds{1}$};
    \draw (2,0) node (40) {$\robotred{$b$}$};

    \draw (00) -- (01);
    \draw (01) -- (02);
    \draw (02) -- (03);
    \draw (03) -- (04);
    \draw (04) -- (05);
    \draw (02) -- (12);
    \draw (12) -- (22);
    \draw (01) -- (11);
    \draw (11) -- (21);
    \draw (21) -- (31);
    \draw (00) -- (10);
    \draw (10) -- (20);
    \draw (20) -- (30);
    \draw (30) -- (40);
\end{tikzpicture}
\end{minipage}
\end{tabular}
\caption{A pathological non-convex discrete example. Different from grid spaces, nodes are connected by edges. Left: Voronoi partition. Middle: A move of agent $a$ changes 
 partitions outside of the neighborhood. Right: Changes in $\mc{P}$ are confined inside the neighborhood.}
\label{fig:externality}
\end{figure}

\color{black}
\camera{In both cases, agents allocate the partition to avoid collisions after they redevide the neighborhood. 
Note that the algorithm updates the partition $\mc{P}_j$ of only neighbors $(i,j) \in E$ and does not affect the other agents outside of $(i,j)$ (Figure \ref{fig:externality}). 
This avoids issues with the convergence of the algorithm in non-convex discrete settings, as  discussed in \cite{yun_rus_2014}.}

\begin{figure}
\begin{minipage}{0.25\textwidth}
\centering
\tikzstyle{background grid}=[draw, black!20,step=.5cm]
\begin{tikzpicture}[show background grid]
\draw[black!20] (0,0) grid (4,2.5);
            \draw (0,0) node {$\robotviolet{$a$}$}; 
            \draw (0.5,0) node {$\robotorange{$b$}$};
            \draw (1,0) node {$\circledorange{1}$};
            \draw (1.5,0) node {$\circledblue{3}$};
            \draw (2,0) node {$\circledblue{2}$};
            \draw (2.5,0) node {$\circled{3}$};
            \draw (3,0) node {$\circledblue{4}$};
            \draw (3.5,0) node {$\circledfilled{2}$};
            \draw (4,0) node {$\circledbrown{1}$};
            \draw (2,0.5) node {$\circledblue{1}$};
            \draw (2.5,0.5) node {$\circledblue{2}$};
            \draw (3,0.5) node {$\circledfilled{2}$};
            \draw (3.5,0.5) node {$\circledfilled{1}$};
            \draw (4,0.5) node {$\robotbrown{$f$}$};
            \draw (2,1) node {$\robotblue{$c$}$};
            \draw (2.5,1) node {$\circledblue{1}$};
            \draw (3,1) node {$\circledfilled{1}$};
            \draw (3.5,1) node {$\robotgreen{$e$}$};
            \draw (4,1) node {$\circledfilled{1}$};
            \draw (2,1.5) node {$\circledblue{1}$};
            \draw (2.5,1.5) node {$\circledred{1}$};
            \draw (3.5,1.5) node {$\circledfilled{1}$};
            \draw (4,1.5) node {$\circledfilled{2}$};
            \draw (2,2) node {$\circledred{1}$};
            \draw (2.5,2) node {$\robotred{$d$}$};
            \draw (3,2) node {$\circledred{1}$};
            \draw (3.5,2) node {$\circledfilled{2}$};
            \draw (4,2) node {$\circledfilled{3}$};
            \draw (2,2.5) node {$\circledred{2}$};
            \draw (2.5,2.5) node {$\circledred{1}$};
            \draw (3,2.5) node {$\circledred{2}$};
            \draw (3.5,2.5) node {$\circledfilled{3}$};
            \draw (4,2.5) node {$\circledfilled{4}$};
            \foreach \x in {0.5,1,1.5,2,2.5}
            \draw (0,\x) node {$\zoneviolet$};
            \foreach \x in {0.5,1,1.5,2,2.5}
            \draw (0.5,\x) node {$\zoneorange$};
            \foreach \x in {0.5,1}
            \draw (1,\x) node {$\zoneorange$};
            \foreach \x in {0.5,1,1.5}
            \draw (1.5,\x) node {$\zoneblue$};
            \foreach \x in {2,2.5}
            \draw (1.5,\x) node {$\zonered$};
            \foreach \x in {1.5,2,2.5}
            \draw (1,\x) node {$\zonered$};
            \draw (3,1.5) node {$\zonered$};
\end{tikzpicture}
\end{minipage}
\caption{After step b, $c$ and $d$ have reallocated themselves in $\mc{P}_c \cup \mc{P}_d$ assuming a third agent would join, leaving the empty position (unfilled) as close as possible from $a$ ($i_{\min}$)}
\label{fig:stepb}
\end{figure}
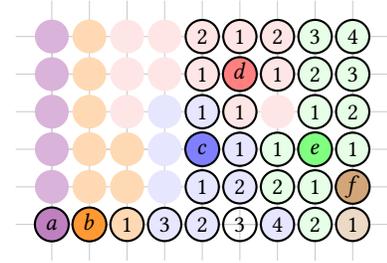

\camera{Let us illustrate this step on our example (Figure \ref{fig:stepb}). For instance, assuming $d$ is picked as the $i^+_{\max}$ agent, this agent interacts with $c$, her parent in the communication graph. Upon evaluating the optimal partition with a third agent, they assess that the gain is higher than $u_{\min}$. As a result, they move and make room for a third agent. As discussed before, note that the region is not updated with respect to that of agent $b$ at this stage. 
}

\color{\newcolor}

Due to the space limitation, the complete proof of Theorem \ref{thm:main} is in the supplementary material~\cite{iwase2024limited}. 
\camera{It follows the line of reasoning described in \cite{marden2016role} but adapted to our setting}. 
Briefly, it proves that the following potential function always increases for each iteration of the algorithm. Since the solution space is finite or compact, then the algorithm terminates.

\begin{equation}
\left.
\begin{array}{l}
\phi(x,\mc{P})=\underset{\substack{i\in\mc{N}}}{\sum}u_i(x,\mc{P})+[V(x,\mc{P})-u_{\min}(x,\mc{P})]_.
\end{array}
\right.
\label{eq:pote}
\end{equation}

\aamasoct{Also, the proof sketch of the approximation ratio is as follows.} Because of the submodularity of $G$ and Equation~ \ref{eq:Z3}, adding agents in optimal allocation $x^*$ to the same number of agents allocated by the algorithm does not make $G$ double. Formally, $G(x^*)\leq G(x,x^*) \leq 2G(x)$.


\subsection{Special Case}

\camera{It is interesting to observe that the algorithm provides an optimal solution in the special case where there are exactly as many points of interest as agents.}  
Let $\mc{C}_+\subseteq C$ be a set of all important points such that $v_c=1,~\forall c\in \mc{C}_+$. Then other points are less important as $v_{c'}\ll 1,~\forall c'\in C\setminus\mc{C}_+$. The algorithm guarantees an optimal solution when agents can cover all the important points as follows.

\begin{theorem}
If $N=|\mc{C}_+|$, Algorithm \ref{alg:extmarden} converges to an optimal solution.
\label{thm:optformation}
\end{theorem}

\begin{proof}
Note that each agent is allocated to a point $c\in \mc{C}_+$ in an optimal solution. Now let us assume that the algorithm converges to a sub-optimal solution $(x,\mc{P})$. In this case, some agents including $i_{\min}$ do not have any points $c\in\mc{C}_+$ in their partitions, due to the neighborhood optimality. Then, there is at least one agent $i$ whose partition $\mc{P}_i$ includes more than two points in $\mc{C}_+$. This violates the condition (\ref{eq:Z2}), which must be satisfied in the convergent state $Z_4$. This is a contradiction.
\end{proof}

\camera{Even though the problem itself is no longer difficult in that case, it is noteworthy that the strategy employed in our algorithm guarantees optimality.}

\subsection{Communication Requirements}
\label{sec:complexity}

As shown in Theorem \ref{thm:main}, the algorithm achieves the approximation ratio of 2, by communicating the minimum degree of information, $u_{\min}$,  $i_{\max}^+$ and $x_{i_{\min}}$. The agent with the maximal gain from an additional agent in her region $i_{\max}^+$ is used to focus on the area to be reallocated in Algorithm \ref{alg:extmarden}. \camera{It is convenient to use it, but we note that it is not strictly required to make the algorithm work (see for instance \cite{marden2016role}).}
The minimum utility, on the other hand, $u_{\min}$ is the key to proving the approximation ratio based on Equation~\ref{eq:Z3}. The location $x_{i_{\min}}$ is required to extend the algorithm for 1-dimensional setting in \cite{marden2016role} to more general settings. 

As for the communication complexity of the algorithm, we start by bounding the convergence rate, which is the number of iterations before convergence. Let $d_{\max}$ be the upper bound of the distance between any two points in the environment $C$, and $\Delta v$ be the resolution limit of the weight $v_c$. In the case of discrete settings, the potential function $\phi(x,\mc{P})$ consists of the elemental term $v_c~ g(|x_i-c|)$ and then the improvement in the potential function for each iteration is lower bounded by $\epsilon=\Delta v \cdot g(d_{\max})$. (In the case of continuous settings, we can regard $\epsilon$ as the agents' resolution limit of utility). Note that the convergence of the algorithm is guaranteed by Theorem \ref{thm:main}. For each iteration, Algorithm \ref{alg:extmarden} checks if the potential function can be improved or not. Before convergence, at least one out of $n$ agents improves the potential function. Then the convergence rate $\alpha$ is bounded as $\alpha \leq n[\phi(x^*,\mc{P}^*)-\phi(x,\mc{P})]/\epsilon$, where $(x,\mc{P})$ and $(x^*,\mc{P}^*)$ are the initial allocation and the allocation after convergence, respectively.

Furthermore, for each iteration, the agents build the communication tree, share the global information, share their private information to compute a neighborhood optimum, and finally share the neighborhood optimum solution. 
\camera{Depending on the exact algorithm used to build the communication tree, the communication burden may vary, but it can be done in $\mc{O}(n^2)$.}  
The communications to share the global information requires at most $\mc{O}(n^2)$ messages. 


\section{Experiments}
\label{sec:exps}

In order to validate the practical efficiency and scalability of our approach, we ran several simulations. 
First, we evaluate the efficiency with small graphs, then we evaluate the scalability with larger graphs. 

In what follows, the nodes in an environment graph are classified into two groups, which are $c\in\mc{C}_+$ with $v_c=1$ and $c'\in C\setminus\mc{C}_+$ with $v_{c'}\ll 1$. Nodes in $\mc{C}_+$ and the initial position of agents are allocated uniformly at random in the environment graph. 
We run 32 simulations for each experiment. 
All the error bars in the figures show 95\% confidence intervals. Note that the environment can be any dimensional space, even though all the graphs are projected into 2D figures. All the numerical results are summarized in Table \ref{tab:results}. As for the implementation, we use Python 3.8.12, Red Hat Enterprise Linux Server release 7.9 and Intel Xeon CPU E5-2670 (2.60 GHz), 192 GB memory to run the experiments. The random number generator is initialized by numpy.random.seed at the beginning of the main code, with different seeds for each run of the simulation.

\paragraph{Comparison.}
The neighborhood optimal approach proposed in Section \ref{sec:mardenext}, coined in the following as NBO, is compared to the following algorithms: 

\begin{itemize}
    \item (VVP) the vanilla distributed covering algorithm based on Voronoi partitioning, as described in Section \ref{sec:distcontrol}. 
    \item (SOTA) the algorithm of Sadeghi \textit{et al.} \cite{sadeghi2020approximation}. 
    In a nutshell, for a given agent $i$, the algorithm first tries to perform individual moves to maximize the (local) social welfare of her neighborhood, 
   \camera{as defined by the Voronoi partitioning}. 
    If no such move is improving, it considers coordinated moves with a single other agent $j$, in the sense that $i$ would move and $j$ would take the place of agent $i$. The algorithm first considers neighbors, and then (via neighborhood communication), may consider agents further away. However, these coordinated moves only involve two agents at most. 
    \item (CGR) the centralized greedy algorithm that allocates agents one by one starting from the empty environment. Recall that this algorithm guarantees a performance of $1-1/e\approx$ 63\% of the optimal solution.
\end{itemize}
    
\aamasoct{We evaluate the performance of the algorithms above with different shapes of the environments (Figure \ref{fig:shapes}). In addition to these shapes, we also use the small bridge setting and OR library dataset shown in \cite{yun_rus_2014}. More details about the shapes and the way instances are generated are available in the supplementary material.}

\begin{figure}[t]
\centering
\includegraphics[width=1\linewidth]{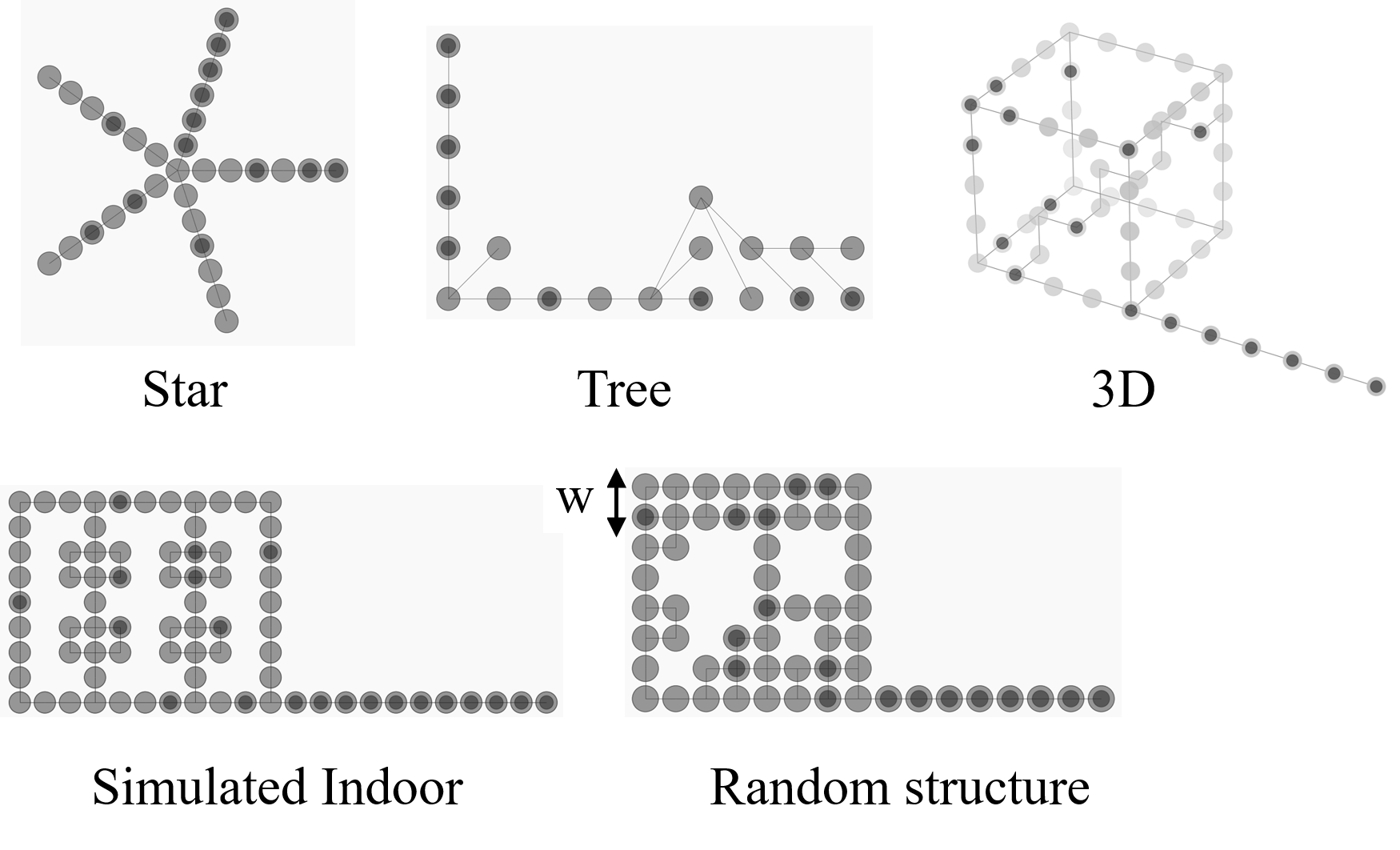}
\vspace{-0.3cm}
\caption{Different shapes of the environment.}
\label{fig:shapes}
\end{figure}

\begin{figure}[t]
\centering
\includegraphics[width=1\linewidth]{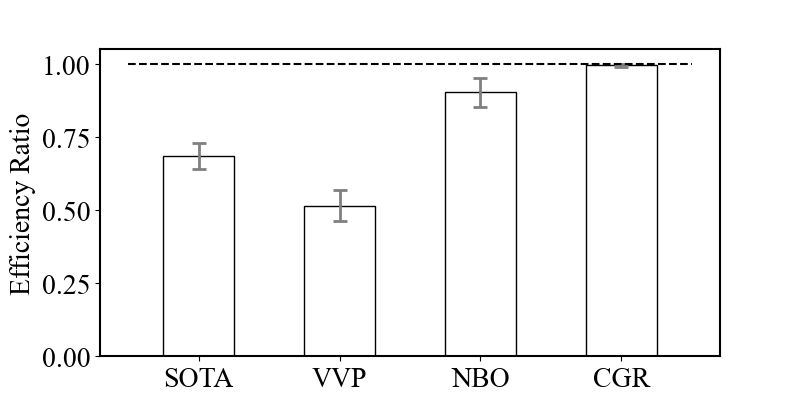}
\vspace{-0.3cm}
\caption{Efficiency ratio $G(x)/G(x^*)$ where $x$ is a solution of each algorithm and $x^*$ is an optimal solution, with $N=5$.}
\label{fig:poa}
\end{figure}



\subsection{Efficiency}

First, we evaluate the efficiency of the proposed neighborhood optimum algorithm, by comparing its solution  with the optimal one and with solutions returned by VVP, SOTA and CGR.
The efficiency is measured as the ratio $G(x)/G(x^*)$ where $x$ is a solution of the algorithm and $x^*$ is an optimal solution. Since finding an optimal solution is NP-hard, the simulation is conducted on a 1D chain with $|C|=20, |\mc{C}_+|=10$, and $N=5$. 

In that case, the proposed algorithm outperforms both distributed algorithms (VVP and SOTA) and improves the efficiency by about 21.7\%  compared to SOTA (Figure \ref{fig:poa}). 
The efficiency ratio of NBO is 90\% which is much better than the theoretical approximation ratio of 2. 
We examine the detail with an illustrative example in the supplementary material.

Due to the NP-hardness of the problem, it is difficult to compare with optimal solutions in larger settings. 
For that reason, we instead compute the efficiency ratio with the centralized greedy (CGR) as a practical baseline. 
Table \ref{tab:results} summarizes all the results obtained. The quantitative comparison of the efficiency shows similar results as in Figure \ref{fig:poa} that the proposed method outperforms the benchmarks in all cases. Note that NBO even sometimes outperforms CGR in some instances, despite the following two disadvantages of NBO compared to CGR: \camera{(i) NBO uses only limited information while CGR uses full information, (ii) NBO starts from a random initial state, which is less favorable compared to CGR's iterative approach from the empty environment.}

\begin{table}
    \centering
    \begin{tabular}{|c|c|c|c|}
        \hline
         Shapes & SOTA & VVP & NBO \\
         \hline 
         \hline
         1D Chains & 0.686 ± 0.045 & 0.515 ± 0.052 & 0.903 ± 0.050 \\
         Stars & 0.983 ± 0.008 & 0.962 ± 0.010 & 1.017 ± 0.006 \\ 
         Trees & 0.952 ± 0.009 & 0.935 ± 0.009 & 0.980 ± 0.004 \\ 
         Simualed indoor & 0.895 ± 0.007 & 0.841 ± 0.010 & 0.917 ± 0.013 \\
         Random ($w=1$) & 0.932 ± 0.010 & 0.904 ± 0.008 & 0.963 ± 0.009 \\ 
         Random  ($w=2$) & 0.936 ± 0.008 & 0.901 ± 0.007 & 0.976 ± 0.009 \\ 
         Small bridge & 0.975 ± 0.007 & 0.955 ± 0.012 & 1.000 ± 0.002 \\
         3D structure & 0.960 ± 0.006 & 0.937 ± 0.007 & 0.982 ± 0.010 \\ 
         OR library & 0.843 ± 0.126 & 0.963 ± 0.014 & 0.978 ± 0.006 \\ 
         \hline         
    \end{tabular}
    \caption{Comparison on different shapes: Efficiency ratio $G(x)/G(x^{CGR})$ where $x$ is a solution of each algorithm and $x^{CGR}$ is a solution of CGR. Each cell shows mean ± standard deviation. Larger values mean better results.}
    \label{tab:results}
\end{table}

\subsection{Scalability}

We finally evaluate the scalability of the proposed algorithm by varying the size of the space $|C|$ and the number of agents $n$ \camera{(we could run experiments with more than 150 agents, the details can be found in the supplementary material).} 
The results show that:
\begin{itemize}
    \item The runtime grows with the size $|C|$, as expected. \camera{On the largest instances of size $|C|=192$, with 20 agents, the runtime is above 3 hours}. 
    \item The runtime decreases when $n$ increases because the size of partitions $|\mc{P}_{ij}|$ also decreases, \camera{and thus the opportunities of improvement are more severely constrained}. 
\end{itemize}

\camera{
These results demonstrate the applicability of the proposed algorithm to real-world medium-scale problems. For larger-scale settings, further improvements will be required to make the approach viable. It should be emphasized though that the anytime nature of the algorithm makes it relevant even with a limited time budget.}
Further details about the experiments are deferred to the supplementary material.

\section{Conclusion}

This paper extends the approach of \cite{marden2016role} to general dimensional settings for discrete environments, ensuring convergence to a neighborhood optimum solution, even in challenging non-convex scenarios, with an approximation ratio of 2. The communication requirements involve disseminating the value and position of the minimum utility agent, capturing an approximation ratio that surpasses other methods presented in \cite{sadeghi2020approximation,durham2011discrete}.  Interestingly, this minimal level of informational dissemination recaptures the best achievable approximation ratio of 2 for Nash equilibria, which requires all agents to have full information about the environment (e.g., choices of all agents, utility associated with all feasible choices, etc.). While more communication demanding than simple best responses based on Voronoi partitioning, the informational dissemination requirements are manageable, enabling local decision-making rules. Furthermore, our algorithm guarantees optimality in a special subclass of the coverage problem and outperforms state-of-the-art benchmark algorithms in experiments. Future research will focus on improving the algorithm's efficiency by minimizing information transmission or enhancing experimental results. Implementing this algorithm on real robotic systems is another important avenue for further exploration.

\begin{acks}
This work is partially supported by ONR grant \#N00014-20-1-2359 and AFOSR grants \#FA9550-20-1-0054 and \#FA9550-21-1-0203.
\end{acks} 


\bibliographystyle{ACM-Reference-Format} 
\bibliography{ref}


\begin{thebibliography}{24}


\ifx \showCODEN    \undefined \def \showCODEN     #1{\unskip}     \fi
\ifx \showDOI      \undefined \def \showDOI       #1{#1}\fi
\ifx \showISBNx    \undefined \def \showISBNx     #1{\unskip}     \fi
\ifx \showISBNxiii \undefined \def \showISBNxiii  #1{\unskip}     \fi
\ifx \showISSN     \undefined \def \showISSN      #1{\unskip}     \fi
\ifx \showLCCN     \undefined \def \showLCCN      #1{\unskip}     \fi
\ifx \shownote     \undefined \def \shownote      #1{#1}          \fi
\ifx \showarticletitle \undefined \def \showarticletitle #1{#1}   \fi
\ifx \showURL      \undefined \def \showURL       {\relax}        \fi
\providecommand\bibfield[2]{#2}
\providecommand\bibinfo[2]{#2}
\providecommand\natexlab[1]{#1}
\providecommand\showeprint[2][]{arXiv:#2}

\bibitem[\protect\citeauthoryear{Arslan, Marden, and Shamma}{Arslan
  et~al\mbox{.}}{2007}]%
        {arslan2007autonomous}
\bibfield{author}{\bibinfo{person}{Gurdal Arslan}, \bibinfo{person}{Jason~R
  Marden}, {and} \bibinfo{person}{Jeff~S Shamma}.}
  \bibinfo{year}{2007}\natexlab{}.
\newblock \showarticletitle{Autonomous vehicle-target assignment: A
  game-theoretical formulation}.
\newblock \bibinfo{journal}{\emph{Journal of Dynamic Systems, Measurement, and
  Control}} \bibinfo{volume}{129}, \bibinfo{number}{5} (\bibinfo{year}{2007}),
  \bibinfo{pages}{584--596}.
\newblock


\bibitem[\protect\citeauthoryear{Beasley}{Beasley}{1990}]%
        {beasley1990or}
\bibfield{author}{\bibinfo{person}{John~E Beasley}.}
  \bibinfo{year}{1990}\natexlab{}.
\newblock \showarticletitle{OR-Library: distributing test problems by
  electronic mail}.
\newblock \bibinfo{journal}{\emph{Journal of the operational research society}}
  \bibinfo{volume}{41}, \bibinfo{number}{11} (\bibinfo{year}{1990}),
  \bibinfo{pages}{1069--1072}.
\newblock


\bibitem[\protect\citeauthoryear{Calinescu, Chekuri, Pal, and
  Vondr{\'a}k}{Calinescu et~al\mbox{.}}{2011}]%
        {calinescu2011maximizing}
\bibfield{author}{\bibinfo{person}{Gruia Calinescu}, \bibinfo{person}{Chandra
  Chekuri}, \bibinfo{person}{Martin Pal}, {and} \bibinfo{person}{Jan
  Vondr{\'a}k}.} \bibinfo{year}{2011}\natexlab{}.
\newblock \showarticletitle{Maximizing a monotone submodular function subject
  to a matroid constraint}.
\newblock \bibinfo{journal}{\emph{SIAM J. Comput.}} \bibinfo{volume}{40},
  \bibinfo{number}{6} (\bibinfo{year}{2011}), \bibinfo{pages}{1740--1766}.
\newblock


\bibitem[\protect\citeauthoryear{Cortes, Martinez, Karatas, and Bullo}{Cortes
  et~al\mbox{.}}{2004}]%
        {cortes2004coverage}
\bibfield{author}{\bibinfo{person}{Jorge Cortes}, \bibinfo{person}{Sonia
  Martinez}, \bibinfo{person}{Timur Karatas}, {and} \bibinfo{person}{Francesco
  Bullo}.} \bibinfo{year}{2004}\natexlab{}.
\newblock \showarticletitle{Coverage control for mobile sensing networks}.
\newblock \bibinfo{journal}{\emph{IEEE Transactions on robotics and
  Automation}} \bibinfo{volume}{20}, \bibinfo{number}{2}
  (\bibinfo{year}{2004}), \bibinfo{pages}{243--255}.
\newblock


\bibitem[\protect\citeauthoryear{Crouse}{Crouse}{2016}]%
        {crouse2016implementing}
\bibfield{author}{\bibinfo{person}{David~F Crouse}.}
  \bibinfo{year}{2016}\natexlab{}.
\newblock \showarticletitle{On implementing 2D rectangular assignment
  algorithms}.
\newblock \bibinfo{journal}{\emph{IEEE Trans. Aerospace Electron. Systems}}
  \bibinfo{volume}{52}, \bibinfo{number}{4} (\bibinfo{year}{2016}),
  \bibinfo{pages}{1679--1696}.
\newblock


\bibitem[\protect\citeauthoryear{Douchan, Wolf, and Kaminka}{Douchan
  et~al\mbox{.}}{2019}]%
        {douchan2019}
\bibfield{author}{\bibinfo{person}{Yinon Douchan}, \bibinfo{person}{Ran Wolf},
  {and} \bibinfo{person}{Gal~A Kaminka}.} \bibinfo{year}{2019}\natexlab{}.
\newblock \showarticletitle{Swarms Can be Rational}. In
  \bibinfo{booktitle}{\emph{Proceedings of the 18th International Conference on
  Autonomous Agents and MultiAgent Systems}}. \bibinfo{pages}{149--157}.
\newblock


\bibitem[\protect\citeauthoryear{Durham, Carli, Frasca, and Bullo}{Durham
  et~al\mbox{.}}{2011}]%
        {durham2011discrete}
\bibfield{author}{\bibinfo{person}{Joseph~W Durham}, \bibinfo{person}{Ruggero
  Carli}, \bibinfo{person}{Paolo Frasca}, {and} \bibinfo{person}{Francesco
  Bullo}.} \bibinfo{year}{2011}\natexlab{}.
\newblock \showarticletitle{Discrete partitioning and coverage control for
  gossiping robots}.
\newblock \bibinfo{journal}{\emph{IEEE Transactions on Robotics}}
  \bibinfo{volume}{28}, \bibinfo{number}{2} (\bibinfo{year}{2011}),
  \bibinfo{pages}{364--378}.
\newblock


\bibitem[\protect\citeauthoryear{Galceran and Carreras}{Galceran and
  Carreras}{2013}]%
        {galceran2013survey}
\bibfield{author}{\bibinfo{person}{Enric Galceran} {and} \bibinfo{person}{Marc
  Carreras}.} \bibinfo{year}{2013}\natexlab{}.
\newblock \showarticletitle{A survey on coverage path planning for robotics}.
\newblock \bibinfo{journal}{\emph{Robotics and Autonomous systems}}
  \bibinfo{volume}{61}, \bibinfo{number}{12} (\bibinfo{year}{2013}),
  \bibinfo{pages}{1258--1276}.
\newblock


\bibitem[\protect\citeauthoryear{Gao, Cort{\'e}s, and Bullo}{Gao
  et~al\mbox{.}}{2008}]%
        {gao2008notes}
\bibfield{author}{\bibinfo{person}{Chunkai Gao}, \bibinfo{person}{Jorge
  Cort{\'e}s}, {and} \bibinfo{person}{Francesco Bullo}.}
  \bibinfo{year}{2008}\natexlab{}.
\newblock \showarticletitle{Notes on averaging over acyclic digraphs and
  discrete coverage control}.
\newblock \bibinfo{journal}{\emph{Automatica}} \bibinfo{volume}{44},
  \bibinfo{number}{8} (\bibinfo{year}{2008}), \bibinfo{pages}{2120--2127}.
\newblock


\bibitem[\protect\citeauthoryear{Haumann, Breitenmoser, Willert, Listmann, and
  Siegwart}{Haumann et~al\mbox{.}}{2011}]%
        {haumann2011discoverage}
\bibfield{author}{\bibinfo{person}{Dominik Haumann}, \bibinfo{person}{Andreas
  Breitenmoser}, \bibinfo{person}{Volker Willert}, \bibinfo{person}{Kim
  Listmann}, {and} \bibinfo{person}{Roland Siegwart}.}
  \bibinfo{year}{2011}\natexlab{}.
\newblock \showarticletitle{DisCoverage for non-convex environments with
  arbitrary obstacles}. In \bibinfo{booktitle}{\emph{2011 IEEE International
  Conference on Robotics and Automation}}. IEEE, \bibinfo{pages}{4486--4491}.
\newblock


\bibitem[\protect\citeauthoryear{Iwase, Beynier, Bredeche, Maudet, and
  Marden}{Iwase et~al\mbox{.}}{2024}]%
        {iwase2024limited}
\bibfield{author}{\bibinfo{person}{Tatsuya Iwase}, \bibinfo{person}{Aur\'{e}lie
  Beynier}, \bibinfo{person}{Nicolas Bredeche}, \bibinfo{person}{Nicolas
  Maudet}, {and} \bibinfo{person}{Jason Marden}.}
  \bibinfo{year}{2024}\natexlab{}.
\newblock \showarticletitle{Is Limited Information Enough? An Approximate
  Multi-agent Coverage Control in Non-Convex Discrete Environments}.
\newblock \bibinfo{journal}{\emph{arXiv preprint arXiv:2401.03752}}
  (\bibinfo{year}{2024}).
\newblock


\bibitem[\protect\citeauthoryear{Koutsoupias and Papadimitriou}{Koutsoupias and
  Papadimitriou}{1999}]%
        {KoutsoupiasEtAl2009}
\bibfield{author}{\bibinfo{person}{Elias Koutsoupias} {and}
  \bibinfo{person}{Christos Papadimitriou}.} \bibinfo{year}{1999}\natexlab{}.
\newblock \showarticletitle{Worst-Case Equilibria}. In
  \bibinfo{booktitle}{\emph{Proceedings of the 16th Symposium on Theoretical
  Aspects of Computer Science, STACS}},
  \bibfield{editor}{\bibinfo{person}{Christoph Meinel} {and}
  \bibinfo{person}{Sophie Tison}} (Eds.). \bibinfo{pages}{404--413}.
\newblock


\bibitem[\protect\citeauthoryear{Marden}{Marden}{2016}]%
        {marden2016role}
\bibfield{author}{\bibinfo{person}{Jason~R Marden}.}
  \bibinfo{year}{2016}\natexlab{}.
\newblock \showarticletitle{The role of information in distributed resource
  allocation}.
\newblock \bibinfo{journal}{\emph{IEEE Transactions on Control of Network
  Systems}} \bibinfo{volume}{4}, \bibinfo{number}{3} (\bibinfo{year}{2016}),
  \bibinfo{pages}{654--664}.
\newblock


\bibitem[\protect\citeauthoryear{Marden and Wierman}{Marden and
  Wierman}{2009}]%
        {marden2009overcoming}
\bibfield{author}{\bibinfo{person}{Jason~R Marden} {and} \bibinfo{person}{Adam
  Wierman}.} \bibinfo{year}{2009}\natexlab{}.
\newblock \showarticletitle{Overcoming limitations of game-theoretic
  distributed control}. In \bibinfo{booktitle}{\emph{Proceedings of the 48h
  IEEE Conference on Decision and Control (CDC) held jointly with 2009 28th
  Chinese Control Conference}}. IEEE, \bibinfo{pages}{6466--6471}.
\newblock


\bibitem[\protect\citeauthoryear{Megiddo and Supowit}{Megiddo and
  Supowit}{1984}]%
        {megiddo1984complexity}
\bibfield{author}{\bibinfo{person}{Nimrod Megiddo} {and}
  \bibinfo{person}{Kenneth~J Supowit}.} \bibinfo{year}{1984}\natexlab{}.
\newblock \showarticletitle{On the complexity of some common geometric location
  problems}.
\newblock \bibinfo{journal}{\emph{SIAM journal on computing}}
  \bibinfo{volume}{13}, \bibinfo{number}{1} (\bibinfo{year}{1984}),
  \bibinfo{pages}{182--196}.
\newblock


\bibitem[\protect\citeauthoryear{Nisan, Roughgarden, Tardos, and
  Vazirani}{Nisan et~al\mbox{.}}{2007}]%
        {nisan2007algorithmic}
\bibfield{author}{\bibinfo{person}{Noam Nisan}, \bibinfo{person}{Tim
  Roughgarden}, \bibinfo{person}{Eva Tardos}, {and} \bibinfo{person}{Vijay~V
  Vazirani}.} \bibinfo{year}{2007}\natexlab{}.
\newblock \showarticletitle{Algorithmic game theory, 2007}.
\newblock \bibinfo{journal}{\emph{Book available for free online}}
  (\bibinfo{year}{2007}).
\newblock


\bibitem[\protect\citeauthoryear{Okumura and D{\'e}fago}{Okumura and
  D{\'e}fago}{2022}]%
        {okumura2022solving}
\bibfield{author}{\bibinfo{person}{Keisuke Okumura} {and}
  \bibinfo{person}{Xavier D{\'e}fago}.} \bibinfo{year}{2022}\natexlab{}.
\newblock \showarticletitle{Solving Simultaneous Target Assignment and Path
  Planning Efficiently with Time-Independent Execution}. In
  \bibinfo{booktitle}{\emph{Proceedings of the International Conference on
  Automated Planning and Scheduling}}, Vol.~\bibinfo{volume}{32}.
  \bibinfo{pages}{270--278}.
\newblock


\bibitem[\protect\citeauthoryear{Sadeghi, Asghar, and Smith}{Sadeghi
  et~al\mbox{.}}{2020}]%
        {sadeghi2020approximation}
\bibfield{author}{\bibinfo{person}{Armin Sadeghi}, \bibinfo{person}{Ahmad~Bilal
  Asghar}, {and} \bibinfo{person}{Stephen~L Smith}.}
  \bibinfo{year}{2020}\natexlab{}.
\newblock \showarticletitle{Approximation algorithms for distributed
  multi-robot coverage in non-convex environments}.
\newblock \bibinfo{journal}{\emph{arXiv preprint arXiv:2005.02471}}
  (\bibinfo{year}{2020}).
\newblock


\bibitem[\protect\citeauthoryear{Shoham and Leyton-Brown}{Shoham and
  Leyton-Brown}{2008}]%
        {Shoham2008}
\bibfield{author}{\bibinfo{person}{Yoav Shoham} {and} \bibinfo{person}{Kevin
  Leyton-Brown}.} \bibinfo{year}{2008}\natexlab{}.
\newblock \bibinfo{booktitle}{\emph{Multiagent systems: {A}lgorithmic,
  Game-Theoretic, and logical foundations}}.
\newblock \bibinfo{publisher}{Cambridge University Press}.
\newblock
\showISBNx{9780511811654}
\showISSN{00796611}
\urldef\tempurl%
\url{https://doi.org/10.1017/CBO9780511811654}
\showDOI{\tempurl}


\bibitem[\protect\citeauthoryear{Stern, Sturtevant, Felner, Koenig, Ma, Walker,
  Li, Atzmon, Cohen, Kumar, et~al\mbox{.}}{Stern et~al\mbox{.}}{2019}]%
        {stern2019multi}
\bibfield{author}{\bibinfo{person}{Roni Stern}, \bibinfo{person}{Nathan~R
  Sturtevant}, \bibinfo{person}{Ariel Felner}, \bibinfo{person}{Sven Koenig},
  \bibinfo{person}{Hang Ma}, \bibinfo{person}{Thayne~T Walker},
  \bibinfo{person}{Jiaoyang Li}, \bibinfo{person}{Dor Atzmon},
  \bibinfo{person}{Liron Cohen}, \bibinfo{person}{TK~Satish Kumar},
  {et~al\mbox{.}}} \bibinfo{year}{2019}\natexlab{}.
\newblock \showarticletitle{Multi-agent pathfinding: Definitions, variants, and
  benchmarks}. In \bibinfo{booktitle}{\emph{Twelfth Annual Symposium on
  Combinatorial Search}}.
\newblock


\bibitem[\protect\citeauthoryear{Wattenhofer}{Wattenhofer}{2016}]%
        {wattenhoffer2016}
\bibfield{author}{\bibinfo{person}{Roger Wattenhofer}.}
  \bibinfo{year}{2016}\natexlab{}.
\newblock \bibinfo{title}{Principles of Distributed Computing}.
  (\bibinfo{year}{2016}).
\newblock
\newblock
\shownote{ETZ Zurich.}


\bibitem[\protect\citeauthoryear{Weiss and Freeman}{Weiss and Freeman}{2001}]%
        {weiss2001optimality}
\bibfield{author}{\bibinfo{person}{Yair Weiss} {and} \bibinfo{person}{William~T
  Freeman}.} \bibinfo{year}{2001}\natexlab{}.
\newblock \showarticletitle{On the optimality of solutions of the max-product
  belief-propagation algorithm in arbitrary graphs}.
\newblock \bibinfo{journal}{\emph{IEEE Transactions on Information Theory}}
  \bibinfo{volume}{47}, \bibinfo{number}{2} (\bibinfo{year}{2001}),
  \bibinfo{pages}{736--744}.
\newblock


\bibitem[\protect\citeauthoryear{Yun and Rus}{Yun and Rus}{2014}]%
        {yun_rus_2014}
\bibfield{author}{\bibinfo{person}{Seung-kook Yun} {and}
  \bibinfo{person}{Daniela Rus}.} \bibinfo{year}{2014}\natexlab{}.
\newblock \showarticletitle{Distributed coverage with mobile robots on a graph:
  locational optimization and equal-mass partitioning}.
\newblock \bibinfo{journal}{\emph{Robotica}} \bibinfo{volume}{32},
  \bibinfo{number}{2} (\bibinfo{year}{2014}), \bibinfo{pages}{257–277}.
\newblock
\urldef\tempurl%
\url{https://doi.org/10.1017/S0263574713001148}
\showDOI{\tempurl}


\bibitem[\protect\citeauthoryear{Zivan, Yedidsion, Okamoto, Glinton, and
  Sycara}{Zivan et~al\mbox{.}}{2015}]%
        {zivan-etal-jaamas15}
\bibfield{author}{\bibinfo{person}{Roie Zivan}, \bibinfo{person}{Harel
  Yedidsion}, \bibinfo{person}{Steven Okamoto}, \bibinfo{person}{Robin
  Glinton}, {and} \bibinfo{person}{Katia Sycara}.}
  \bibinfo{year}{2015}\natexlab{}.
\newblock \showarticletitle{Distributed constraint optimization for teams of
  mobile sensing agents}.
\newblock \bibinfo{journal}{\emph{Autonomous Agents and Multi-Agent Systems}}
  \bibinfo{volume}{29}, \bibinfo{number}{3} (\bibinfo{date}{may}
  \bibinfo{year}{2015}), \bibinfo{pages}{495–536}.
\newblock
\showISSN{1387-2532}
\urldef\tempurl%
\url{https://doi.org/10.1007/s10458-014-9255-3}
\showDOI{\tempurl}


\end{thebibliography}

\newpage
\onecolumn
\appendix

\section{Proof of Theorem \ref{thm:main} (convergence)}
\label{sec:proofconverge}

Here, a pair of neighboring agents $S=\{i,j\}$ is randomly chosen to update their joint actions. For simplicity, we assume $|S|=2$ to explain, but the following proofs hold even for the case of $|S|>2$. Without loss of generality, we assume that $i$ is either the root of the minimum utility tree or closer than $j$ according to the communication tree.
For each iteration, Algorithm \ref{alg:extmarden} will produce a new state $(x',\mc{P}')$ from an old state $(x,\mc{P})$.  Note that after step a, where agents in $S$ reallocate $\mc{P}_S$ by themselves,
\begin{eqnarray}
\begin{array}{l}
M_2(\emptyset,\mc{P}_S)=\underset{\substack{i\in S}}{\sum}u_i(x',\mc{P}').
\end{array}
\end{eqnarray}

Also, after step b, where agents in $S$ try to make a space in $\mc{P}_S$ to accommodate another agent,
\begin{eqnarray}
\begin{array}{l}
M_3(\emptyset,\mc{P}_S)=\underset{\substack{i\in S}}{\sum}u_i(x',\mc{P}')+M_1(x'_{i_+},\mc{P}'_{i_+}).
\end{array}
\end{eqnarray}

\begin{lemma}
If $(x,\mc{P})\in Z_1$, then Algorithm \ref{alg:extmarden} will produce a sequence of states that results in a new state $(x',\mc{P}')\in Z_2$.
\label{thm:z1z2}
\end{lemma}

\begin{proof}
At first,
\begin{eqnarray}
\begin{array}{ll}
V(x,\mc{P})&= \underset{\substack{i\in\mc{N}}}{\max}\ M_1(x_i;\mc{P}_i)\\
&=M_1(x_{i_{\max}^+};\mc{P}_{i_{\max}^+})\\
&=\underset{\substack{y\in\mc{P}_{i_{\max}^+}}}{\max}G(y,x_{i_{\max}^+};\mc{P}_{i_{\max}^+})-G(x_{i_{\max}^+};\mc{P}_{i_{\max}^+})\\
&=\underset{\substack{y\in\mc{P}_{i_{\max}^+}}}{\max}G(y,x_{i_{\max}^+};\mc{P}_{i_{\max}^+})-u_{i_{\max}^+}(x,\mc{P}).
\end{array}
\end{eqnarray}

In case of step a, if $i_{\min}(x,\mc{P})\in  S$, 

\begin{eqnarray}
\begin{array}{lcl}
\underset{\substack{j\in S}}{\sum}u_j(x,\mc{P})+V(x,\mc{P})-u_{\min}(x,\mc{P}) 
& =& \underset{\substack{j\in S\setminus\{i_{\min}(x,\mc{P})\}}}{\sum}u_j(x,\mc{P})+V(x,\mc{P})\\
&=& \underset{\substack{j\in S\setminus\{i_{\min}(x,\mc{P}),i_{\max}^+(x,\mc{P})\}}}{\sum}u_j(x,\mc{P})+u_{i_{\max}^+}(x,\mc{P})\\
&&+\underset{\substack{y\in\mc{P}_{i_{\max}^+}}}{\max}G(y,x_{i_{\max}^+};\mc{P}_{i_{\max}^+})-u_{i_{\max}^+}(x,\mc{P})\\
&=& \underset{\substack{j\in S\setminus\{i_{\min}(x,\mc{P}),i_{\max}^+(x,\mc{P})\}}}{\sum}u_j(x,\mc{P})+\underset{\substack{y\in\mc{P}_{i_{\max}^+}}}{\max}G(y,x_{i_{\max}^+};\mc{P}_{i_{\max}^+}).
\end{array}
\end{eqnarray}

Then,
\begin{eqnarray}
\begin{array}{l}
\underset{\substack{j\in S}}{\sum}u_j(x',\mc{P}')\\
=M_2(\emptyset,\mc{P}_S)\\
=\underset{\substack{y_1,\ldots,y_k\in\mc{P}_S}}{\max}G(y_1,\ldots,y_k;\mc{P}_S)\\
\geq \underset{\substack{y_1,\ldots,y_{k-2}\in\mc{P}_S\setminus\{i_{\min}(x,\mc{P}),i_{\max}^+(x,\mc{P})\}}}{\max}G(y_1,\ldots,y_{k-2};\mc{P}_S)+\underset{\substack{y\in\mc{P}_{i_{\max}^+}}}{\max}G(y,x_{i_{\max}^+};\mc{P}_{i_{\max}^+})\\
\geq \underset{\substack{j\in S\setminus\{i_{\min}(x,\mc{P}),i_{\max}^+(x,\mc{P})\}}}{\sum}u_j(x,\mc{P})+\underset{\substack{y\in\mc{P}_{i_{\max}^+}}}{\max}G(y,x_{i_{\max}^+};\mc{P}_{i_{\max}^+})\\
=\underset{\substack{j\in S}}{\sum}u_j(x,\mc{P})+V(x,\mc{P})-u_{\min}(x,\mc{P}).
\end{array}
\end{eqnarray}

If $i_{\min}(x,\mc{P})\not\in  S$, since $M_3(\emptyset,\mc{P}_S)-M_2(\emptyset,\mc{P}_S)\leq u_{\min}(x,\mc{P})$,
\begin{eqnarray}
\begin{array}{ll}
\underset{\substack{j\in S}}{\sum}u_j(x',\mc{P}')&=M_2(\emptyset,\mc{P}_S)\\
&\geq M_3(\emptyset,\mc{P}_S)-u_{\min}(x,\mc{P}).
\end{array}
\label{eq:z1-2tmp}
\end{eqnarray}

Since $i_{\max}^+\in S$, $\underset{\substack{j\in S}}{\sum}u_j(x,\mc{P})+V(x,\mc{P})=\underset{\substack{j\in S}}{\sum}u_j(x,\mc{P})+M_1(x_{i_{\max}^+};\mc{P}_{i_{\max}^+})$ means a sum of utilities when $3$ agents share $ S$. Then, 

\begin{eqnarray}
\begin{array}{l}
M_3(\emptyset,\mc{P}_S)\geq\underset{\substack{j\in S}}{\sum}u_j(x,\mc{P})+V(x,\mc{P}).
\end{array}
\end{eqnarray}

With (\ref{eq:z1-2tmp}), we have

\begin{eqnarray}
\begin{array}{l}
\underset{\substack{j\in S}}{\sum}u_j(x',\mc{P}')\geq\underset{\substack{j\in S}}{\sum}u_j(x,\mc{P})+V(x,\mc{P})-u_{\min}(x,\mc{P}).
\end{array}
\end{eqnarray}

Note that $V(x,\mc{P})>u_{\min}(x,\mc{P})$ in $Z_1$.
Therefore, in both cases ($i_{\min}(x,\mc{P})\in  S$ and $i_{\min}(x,\mc{P})\not\in  S$), we have

\begin{eqnarray}
\begin{array}{lll}
\phi(x',\mc{P}')&=&\underset{\substack{i\in\mc{N}}}{\sum}u_i(x',\mc{P}')+[V(x',\mc{P}')-u_{\min}(x',\mc{P}')]_+\\
&\geq& \underset{\substack{i\in\mc{N}}}{\sum}u_i(x,\mc{P})+[V(x,\mc{P})-u_{\min}(x,\mc{P})]_+\\
&&+[V(x',\mc{P}')-u_{\min}(x',\mc{P}')]_+\\
&\geq&\phi(x,\mc{P}).
\end{array}
\end{eqnarray}

Note that the equality holds only when $V(x',\mc{P}')-u_{\min}(x',\mc{P}')\leq 0$, which is $(x',\mc{P}')\in Z_2$. As long as the state stays in $Z_1$, $\phi$ increases.

In case of step b, $i_{\min}(x,\mc{P})\not\in  S$ and  $M_3(\emptyset,\mc{P}_S)-M_2(\emptyset,\mc{P}_S)> u_{\min}(x,\mc{P})$. 

First, step b computes $(x',\mc{P}')$ as if an imaginary agent $l$ is added into $\mc{P}_S$, especially at $\mc{P}'_{i_+}$. Then,
\begin{eqnarray}
\begin{array}{l}
\underset{\substack{i\in S}}{\sum}u_i(x',\mc{P}')=M_3(\emptyset,\mc{P}_S)-M_1(x'_{i_+},\mc{P}'_{i_+})
\end{array}
\end{eqnarray}

Note that
\begin{eqnarray}
\begin{array}{ll}
M_1(x'_{i_+},\mc{P}'_{i_+})\\
=M_3(\emptyset,\mc{P}_S)-\underset{\substack{i\in S}}{\sum}u_i(x',\mc{P}')\\
\geq M_3(\emptyset,\mc{P}_S)-M_2(\emptyset,\mc{P}_S)\\
>u_{\min}(x,\mc{P}).
\end{array}
\label{eq:z1-2tmp2}
\end{eqnarray}

Meanwhile, since $i_{\max}^+\in S$,

\begin{eqnarray}
\begin{array}{l}
\underset{\substack{i\in S}}{\sum}u_i(x,\mc{P})+V(x,\mc{P})\\
=\underset{\substack{i\in S}}{\sum}u_i(x,\mc{P})+M_1(x_{i_{\max}^+};\mc{P}_{i_{\max}^+})\\
\leq M_3(\emptyset,\mc{P}_S)\\
=\underset{\substack{i\in S}}{\sum}u_i(x',\mc{P}')+M_1(x'_{i_+},\mc{P}'_{i_+}).
\end{array}
\label{eq:z1-2tmp3}
\end{eqnarray}




Then, since $i_{\min}(x,\mc{P})\not\in  S$,

\begin{eqnarray}
\begin{array}{lll}
\phi(x',\mc{P}')&=&\underset{\substack{i\in\mc{N}}}{\sum}u_i(x',\mc{P}')+[V(x',\mc{P}')-u_{\min}(x',\mc{P}')]_+\\
&=&\underset{\substack{i\in\mc{N}}}{\sum}u_i(x',\mc{P}')+[V(x',\mc{P}')-u_{\min}(x,\mc{P})]_+.
\end{array}
\end{eqnarray}

By the definition of $V$ and (\ref{eq:z1-2tmp2}),

\begin{eqnarray}
\begin{array}{l}
\geq\underset{\substack{i\in\mc{N}}}{\sum}u_i(x',\mc{P}')+M_1(x'_{i_+},\mc{P}'_{i_+})-u_{\min}(x,\mc{P}).
\end{array}
\end{eqnarray}

By (\ref{eq:z1-2tmp3}),

\begin{eqnarray}
\begin{array}{l}
\geq\underset{\substack{i\in\mc{N}}}{\sum}u_i(x,\mc{P})+V(x,\mc{P})-u_{\min}(x,\mc{P}).
\end{array}
\end{eqnarray}

Since $(x,\mc{P})\in Z_1$,

\begin{eqnarray}
\begin{array}{l}
=\underset{\substack{i\in\mc{N}}}{\sum}u_i(x,\mc{P})+[V(x,\mc{P})-u_{\min}(x,\mc{P})]_+=\phi(x,\mc{P}).
\end{array}
\end{eqnarray}

Note that if $V(x',\mc{P}')>M_1(x'_{i_+},\mc{P}'_{i_+})$, then $\phi(x',\mc{P}')>\phi(x,\mc{P})$. Also, from (\ref{eq:z1-2tmp2}), $(x',\mc{P}')\in Z_1$.

If $V(x',\mc{P}')=M_1(x'_{i_+},\mc{P}'_{i_+})$, it means  $i_+=i_{\max}^+(x',\mc{P}')$. Then, $i_+$ will be in the new neighborhood (say$\mc{N}'_i$) in the next iteration. Since $i_{\min}(x,\mc{P})\not\in  S$ and $i_+$ (the nearest agent to $i_{\min}$) has additional space $\mc{P}_l$ for other agent to come in, Algorithm \ref{alg:extmarden} can repeat this process to make $i_{\min}(x,\mc{P})$ move towards larger space without being stuck.

Then, as long as $(x,\mc{P})\in Z_1$, $\phi(x,\mc{P})$ increases. Since the solution space is finite, $(x,\mc{P})$ reaches $Z_2$ in the end.

\end{proof}

\begin{lemma}
If $(x,\mc{P})\in Z_2\setminus Z_3$, then Algorithm \ref{alg:extmarden} will produce a sequence of states that results in a new state $(x',\mc{P}')\in Z_3$. 
\label{thm:z2z3}
\end{lemma}

\begin{proof}
Since $(x,\mc{P})\in Z_2$, $V(x,\mc{P})\leq u_{\min}(x,\mc{P})$.
In case of step a,
\begin{eqnarray}
\begin{array}{l}
\phi(x',\mc{P}')\\
=\underset{\substack{i\in\mc{N}}}{\sum}u_i(x',\mc{P}')+[V(x',\mc{P}')-u_{\min}(x',\mc{P}')]_+\\
\geq \underset{\substack{i\in\mc{N}}}{\sum}u_i(x',\mc{P}')\\
=M_2(\emptyset,\mc{P}_S)+\underset{\substack{i\in\mc{N}\setminus S}}{\sum}u_i(x,\mc{P})\\
\geq \underset{\substack{i\in\mc{N}}}{\sum}u_i(x,\mc{P})\\
=\underset{\substack{i\in\mc{N}}}{\sum}u_i(x,\mc{P})+[V(x,\mc{P})-u_{\min}(x,\mc{P})]_+\\
=\phi(x,\mc{P}).
\end{array}
\end{eqnarray}

Note that the equality holds only when $M_2(\emptyset,\mc{P}_S)=\underset{\substack{i\in S}}{\sum}u_i(x,\mc{P})$. This means that $\phi$ increases as long as the local state in $ S$ does not satisfy the condition of $Z_4$.

In case of step b,

\begin{eqnarray}
\begin{array}{l}
M_3(\emptyset,\mc{P}_S)\\
=\underset{\substack{i\in S}}{\sum}u_i(x',\mc{P}')+M_1(x'_{jnearest};\mc{P}'_{jnearest})\\
\leq\underset{\substack{i\in S}}{\sum}u_i(x',\mc{P}')+V(x',\mc{P}').
\end{array}
\end{eqnarray}

Since $(x,\mc{P})\in Z_2\setminus Z_3$, $M_3(\emptyset,\mc{P}_S)-M_2(\emptyset,\mc{P}_S)> u_{\min}(x,\mc{P})$. Then,

\begin{eqnarray}
\begin{array}{l}
\underset{\substack{i\in S}}{\sum}u_i(x',\mc{P}')+V(x',\mc{P}')\\
\geq M_3(\emptyset,\mc{P}_S)\\
>M_2(\emptyset,\mc{P}_S)+u_{\min}(x,\mc{P})\\
\geq \underset{\substack{i\in S}}{\sum}u_i(x,\mc{P})+u_{\min}(x,\mc{P}).
\end{array}
\end{eqnarray}

Then, since $i_{\min}(x,\mc{P})\not\in  S$,

\begin{eqnarray}
\begin{array}{l}
\phi(x',\mc{P}')\\
=\underset{\substack{i\in\mc{N}}}{\sum}u_i(x',\mc{P}')+[V(x',\mc{P}')-u_{\min}(x',\mc{P}')]_+\\
=\underset{\substack{i\in\mc{N}}}{\sum}u_i(x',\mc{P}')+[V(x',\mc{P}')-u_{\min}(x,\mc{P})]_+\\
\geq \underset{\substack{i\in\mc{N}}}{\sum}u_i(x',\mc{P}')+V(x',\mc{P}')-u_{\min}(x,\mc{P})\\
>\underset{\substack{i\in\mc{N}}}{\sum}u_i(x,\mc{P})\\
=\underset{\substack{i\in\mc{N}}}{\sum}u_i(x,\mc{P})+[V(x,\mc{P})-u_{\min}(x,\mc{P})]_+\\
=\phi(x,\mc{P}).
\end{array}
\end{eqnarray}

Then $\phi$ always increases as long as $(x,\mc{P})\in Z_2\setminus Z_3$. Note that $(x',\mc{P}')$ can be in $Z_1$. However, according to Lemma \ref{thm:z1z2}, the state will come back to $Z_2$ without cycle. Since $(Z_2\setminus Z_3) \cup Z_1$ is finite, $(x,\mc{P})$ reaches $Z_3$ in the end.
\end{proof}

\begin{lemma}
If $(x,\mc{P})\in Z_3$, then Algorithm \ref{alg:extmarden} will produce a sequence of states that results in a new state $(x',\mc{P}')\in Z_4$.
\label{thm:z3z4}
\end{lemma}

\begin{proof}
By definition,

\begin{eqnarray}
\begin{array}{l}
M_2(\emptyset,\mc{P}_S)\\
=\underset{\substack{y_1,\ldots,y_k\in\mc{P}_S}}{\max}G(y_1,\ldots,y_k;\mc{P}_S)\\
\geq G(x;\mc{P}_S).\\
\geq \underset{\substack{i\in S}}{\sum}u_i(x,\mc{P})\\
\end{array}
\end{eqnarray}

Since $(x,\mc{P})\in Z_3$, Algorithm \ref{alg:extmarden} always runs step a. In this case, $V$ always decreases and $u_{\min}$ always increases. Then $(x,\mc{P}), (x',\mc{P}')\in Z_2$. Then,

\begin{eqnarray}
\begin{array}{l}
\phi(x',\mc{P}')\\
=\underset{\substack{i\in\mc{N}}}{\sum}u_i(x',\mc{P}')+[V(x',\mc{P}')-u_{\min}(x',\mc{P}')]_+\\
=\underset{\substack{i\in\mc{N}}}{\sum}u_i(x',\mc{P}')\\
=M_2(\emptyset,\mc{P}_S)+\underset{\substack{i\in\mc{N}\setminus S}}{\sum}u_i(x,\mc{P})\\
\geq \underset{\substack{i\in\mc{N}}}{\sum}u_i(x,\mc{P})\\
=\phi(x,\mc{P}).
\end{array}
\end{eqnarray}

Note that the equality holds only when $M_2(\emptyset,\mc{P}_S)=\underset{\substack{i\in S}}{\sum}u_i(x,\mc{P})$. This means that $\phi(x,\mc{P})$ increases as long as $(x,\mc{P})\in Z_3$,  until $(x,\mc{P})$ reaches $Z_4$.
\end{proof}

\begin{proof}[Proof of Theorem \ref{thm:main} (Convergence)]
It follows Lemma \ref{thm:z1z2} to \ref{thm:z3z4}.
\end{proof}

Note: The proofs above hold even if we replace $ S$ with the general definition of neighborhood $\mc{N}_i$.

\section{Proof of Theorem \ref{thm:main} (Approximation ratio)}
\label{sec:proofpoa}

\begin{lemma}
If $(x;\mc{P}) \in Z_4$, $M_1(x;\mc{P})\leq u_{\min}(x,\mc{P})$.
\label{thm:umin}
\end{lemma}

\begin{proof}
Since $(x;\mc{P}) \in Z_4\subseteq Z3$, $M_3(\emptyset,\mc{P}_S)-M_2(\emptyset,\mc{P}_S)\leq u_{\min}(x,\mc{P}),~\forall i\in\mc{N}$. Also, $M_2(\emptyset,\mc{P}_S)=\underset{\substack{i\in S}}{\sum}u_i(x,\mc{P})=G(x,\mc{P}_S),~\forall i\in\mc{N}$. 
Then, forall $i\in\mc{N}$,

\begin{eqnarray}
\begin{array}{l}
M_1(x_{\{i,j\}},\mc{P}_S)\\
=\underset{\substack{y\in\mc{P}_S}}{\max}G(y,x_{\{i,j\}};\mc{P}_S)-G(x_{\{i,j\}},\mc{P}_S)\\
=\underset{\substack{y\in\mc{P}_S}}{\max}G(y,x_{\{i,j\}};\mc{P}_S)-\underset{\substack{x_{\{i,j\}}\in\mc{P}_S}}{\max}G(x_{\{i,j\}},\mc{P}_S)\\
\leq\underset{\substack{y_1,\ldots,y_{3}\in\mc{P}_S}}{\max}G(y_1,\ldots,y_{3};\mc{P}_S)-\underset{\substack{y_1,y_2\in\mc{P}_S}}{\max}G(y_1,y_2;\mc{P}_S)\\
=M_3(\emptyset,\mc{P}_S)-M_2(\emptyset,\mc{P}_S)\\
\leq u_{\min}(x,\mc{P}).
\end{array}
\end{eqnarray}

Then, $u_{\min}(x,\mc{P})\geq \underset{\substack{i\in\mc{N}}}{\max}M_1(x_{\{i,j\}};\mc{P}_S)=M_1(x;\mc{P})$.
\end{proof}

\begin{lemma}
$G(x;\mc{C})$ is submodular in $x$. For any allocation $x\in\mc{X}$, agent sets $S\subseteq T\subseteq \mc{N}$, and agent $i\in S$,
\begin{eqnarray}
\begin{array}{l}
G(x_S;\mc{C})-G(x_{S\setminus \{i\}};\mc{C})\geq G(x_T;\mc{C})-G(x_{T\setminus \{i\}};\mc{C}).
\end{array}
\end{eqnarray}
\label{thm:submodG}
\end{lemma}

\begin{proof}
Because $\mc{P}_i(x_S;\mc{C})\subseteq \mc{P}_i(x_T;\mc{C})$. This is because of agents $T-S$.
\end{proof}

\begin{lemma}
$M_k(x,\mc{C})$ is monotone decreasing in $x$.
\begin{eqnarray}
\begin{array}{l}
M_k(x\setminus \{x_i\};\mc{C})\geq M_k(x;\mc{C}), \forall x_i\in\mc{C}.
\end{array}
\end{eqnarray}
\end{lemma}

\begin{proof}
By Lemma \ref{thm:submodG}, 
\begin{eqnarray}
\begin{array}{l}
G(y_1,\ldots,y_k,x\setminus \{x_i\};\mc{C})-G(x\setminus \{x_i\};\mc{C})\geq G(y_1,\ldots,y_k,x;\mc{C})-G(x;\mc{C}),\\ \forall y_1,\ldots,y_k \in \mc{C}.
\end{array}
\end{eqnarray}
\end{proof}

\begin{proof}[Proof of Theorem \ref{thm:poa}]
Let $(x,\mc{P})$ be a solution of Algorithm \ref{alg:extmarden} and $(x^*,\mc{P}^*)$ be an optimal solution. Also, we denote $G(x;\mc{C})$ as $G(x)$. Then,
\begin{eqnarray}
\begin{array}{lll}
G(x^*)&\leq& G(x,x^*)\\
&=&G(x)\\
&&+[G(x,x^*_1)-G(x)]\\
&&+[G(x,x^*_1,x^*_2)-G(x,x^*_1)]\\
&&+\ldots\\
&&+[G(x,x^*)-G(x,x^*_1,\ldots,x^*_{N-1})]\\
&\leq& G(x)\\
&&+M_1(x,\mc{P})\\
&&+M_1(x,x^*_1,\mc{P})\\
&&+\ldots\\
&&+M_1(x,x^*_1,\ldots,x^*_{N-1},\mc{P})\\
&\leq& G(x)+n*M_1(x,\mc{P})\\
&\leq& G(x)+n*u_{\min}(x,\mc{P})\\
&\leq& 2G(x).
\end{array}
\end{eqnarray}
\end{proof}

Note: The proofs above also hold even if we replace $ S$ with the general definition of neighborhood $\mc{N}_i$.

\section{Comparison of algorithms in 1 dimensional setting}
\label{sec:comparison}

We examine the different behavior of algorithms (VVP, SOTA, NBO and the optimal solution, OPT) with simple 1-dimensional setting (Figure \ref{fig:comparison}). Initially, two agents are located at the left end. Agent 0 moves first in all cases. In SOTA, each agent is activated once. The agent first tries to improve the objective function by itself. Then, the dagent also tries to communicate with other agents to improve the objective function together. In B, agent 0 tries to move first, but it can't move because it is blocked by agent 1. Then, agent 0 is deactivated and does not move anymore by itself. Next, agent 1 moves right, stops at its best position, and tries to improve the objective function further by communicating with agent 0. However, the algorithm cannot find any collaborative moves to improve the situation and  terminates. Meanwhile, in C, NBO finds an optimal solution because the two agents are neighbors. In this case, the coverage control reaches another optimal solution as in D.

As shown above, SOTA depends on the action order of agents. Then we ran the same simulation by switching the initial locations of the agents. In E, SOTA can move both agents, and they reach a better solution than B (but not optimal). Meanwhile, NBO can find an optimal solution  in a stable manner as in F.

\begin{figure}[t]
\centering
\includegraphics[width=.6\linewidth]{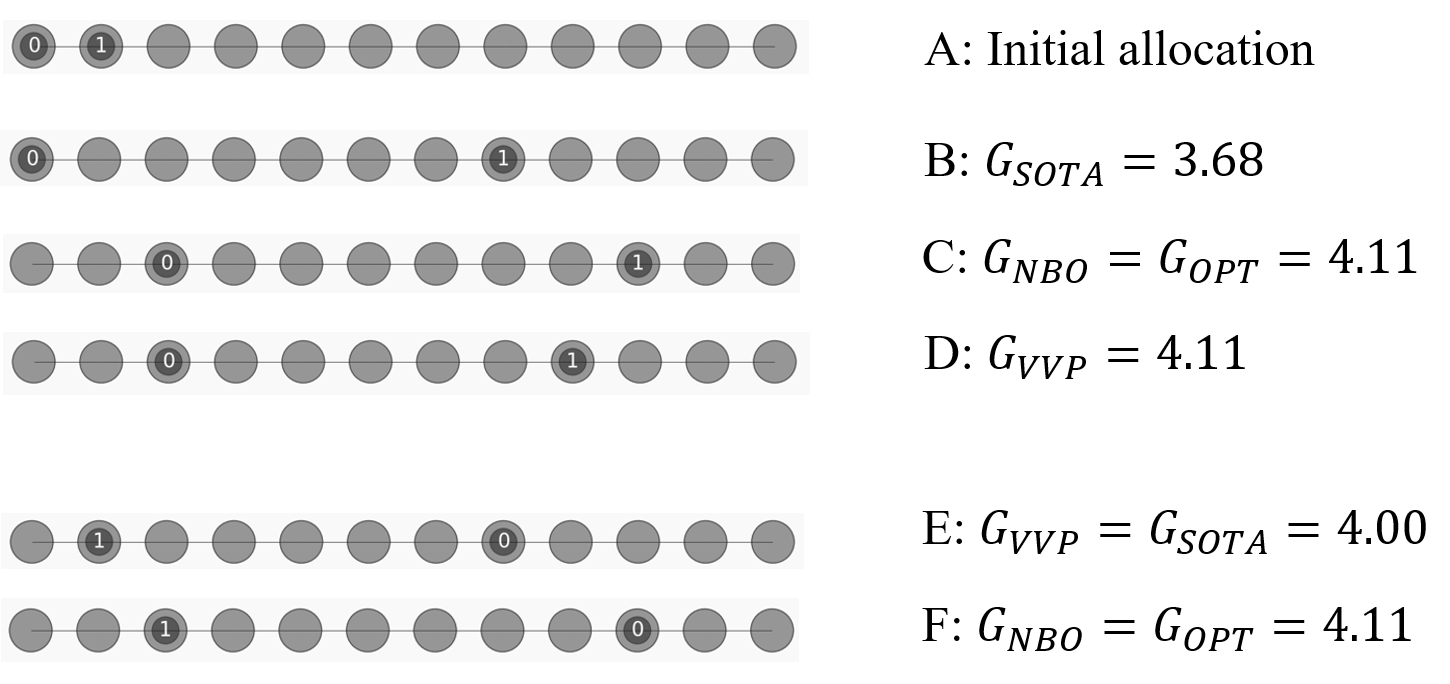}
\caption{1-dimensional setting for comparison of algorithms with $n=2, m=12$. All nodes have the same weight $v_c=1$. Agent 0 moves first in all results. A: The initial allocation. B \ldots D: Outcomes of algorithms. $G_{*}$ is the value of the objective function. $G_{OPT}$ is the optimal value. E,F: Outcomes with switching the initial locations of the agents.}
\label{fig:comparison}
\end{figure}

\section{Details of Experiments}
\label{sec:scalability}

\paragraph{Shapes.}
We ran our experiments on different shapes: 
\begin{itemize}
    \item (1D) \emph{lines} ---While the original approach of Marden \cite{marden2009overcoming} was sufficient on one-dimensional structures, we use it as a benchmark since optimal solutions can be computed when the number of agents remains reasonable.  
    \item (2D) \emph{stars} and \emph{trees} ---We generated two representatives of these standard shapes (Figure \ref{fig:startree}). 
    \item (2D) \emph{simulated indoor environment} ---This is based on a possible application which could consist in building a mesh network in an indoor environment (Figure \ref{fig:indoor}), where coverage could be needed after a catastrophic event (e.g. building inaccessible after a nuclear accident). 
    \item (2D) \emph{randomly generated structures} ---We start from a template structure with the parameter of corridor width $w \in \{1,2\}$ (see Figure~\ref{fig:randmaze}). From such a template structure, a connected structure is then generated by randomly removing nodes. 
    \item (2D) \emph{a small bridge} ---This is an example of non-convex discrete environment shown in \cite{yun_rus_2014}.
    \item (3D) \emph{non-planar graphs} ---To illustrate how the approach can work in higher dimensions, we selected some non-planar graphs (Figure \ref{fig:3d}). 
    \item ($\geq$ 3D) \emph{OR library} ---This is a test dataset consisting of random non-planar graphs for $p$-median problem \cite{beasley1990or}.
\end{itemize}

\aamasoct{Though we also try to evaluate the cooperative version of VVP, where each agent $i$ tries to maximize the social welfare of her neighborhood, we discontinue the experiment due to its lack of convergence in practice. (For example, it fails to converge in more than 60\% of cases in 3D setting). We also omit the evaluation of the algorithm in \cite{yun_rus_2014}, because it can be regarded as a variant of VVP by skipping the moves that could change $\mc{N}_{\mc{N}_i}$. In the cases of OR library, we have results only for 6 instances due to the large problem size. }

\begin{figure}[t]
\centering
\includegraphics[width=.8\linewidth]{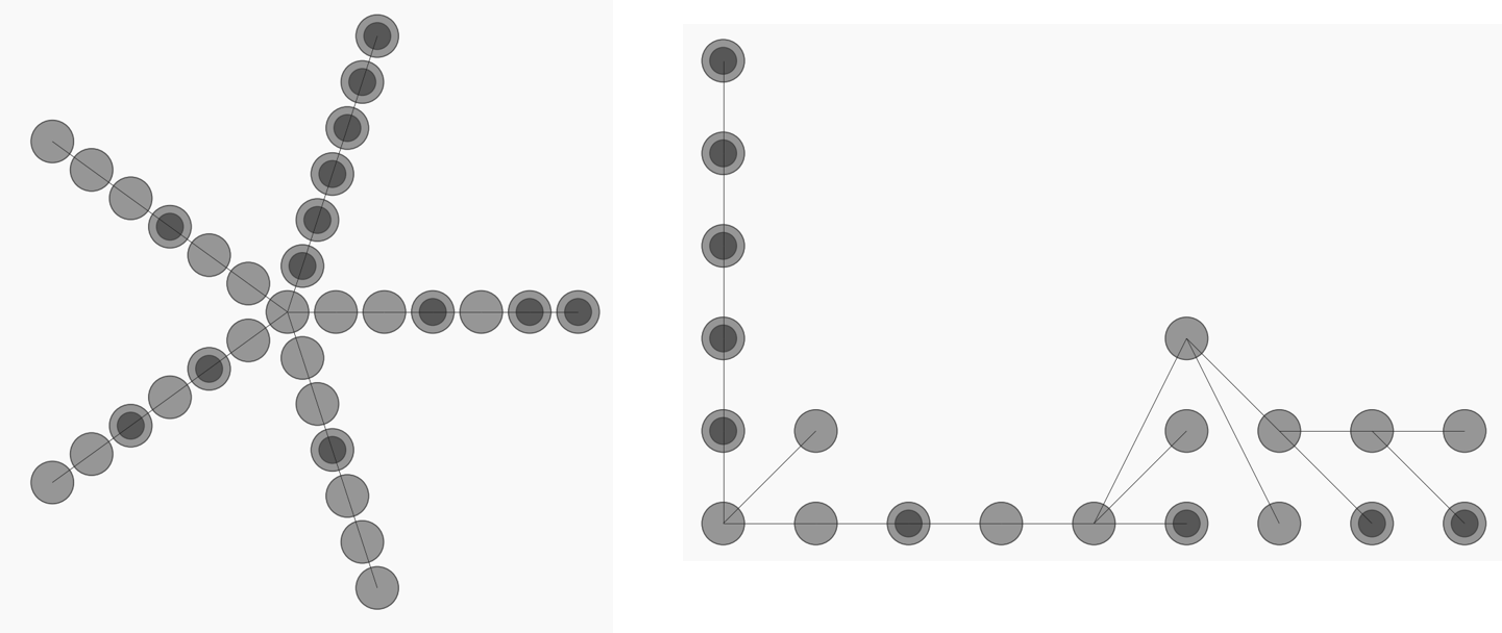}
\caption{Left: A star with extended branches. Right: A tree. Light gray nodes show targets ($v_c=1$), and dark gray nodes show agents.}
\label{fig:startree}
\end{figure}

\begin{figure}[t]
\centering
\includegraphics[width=1.\linewidth]{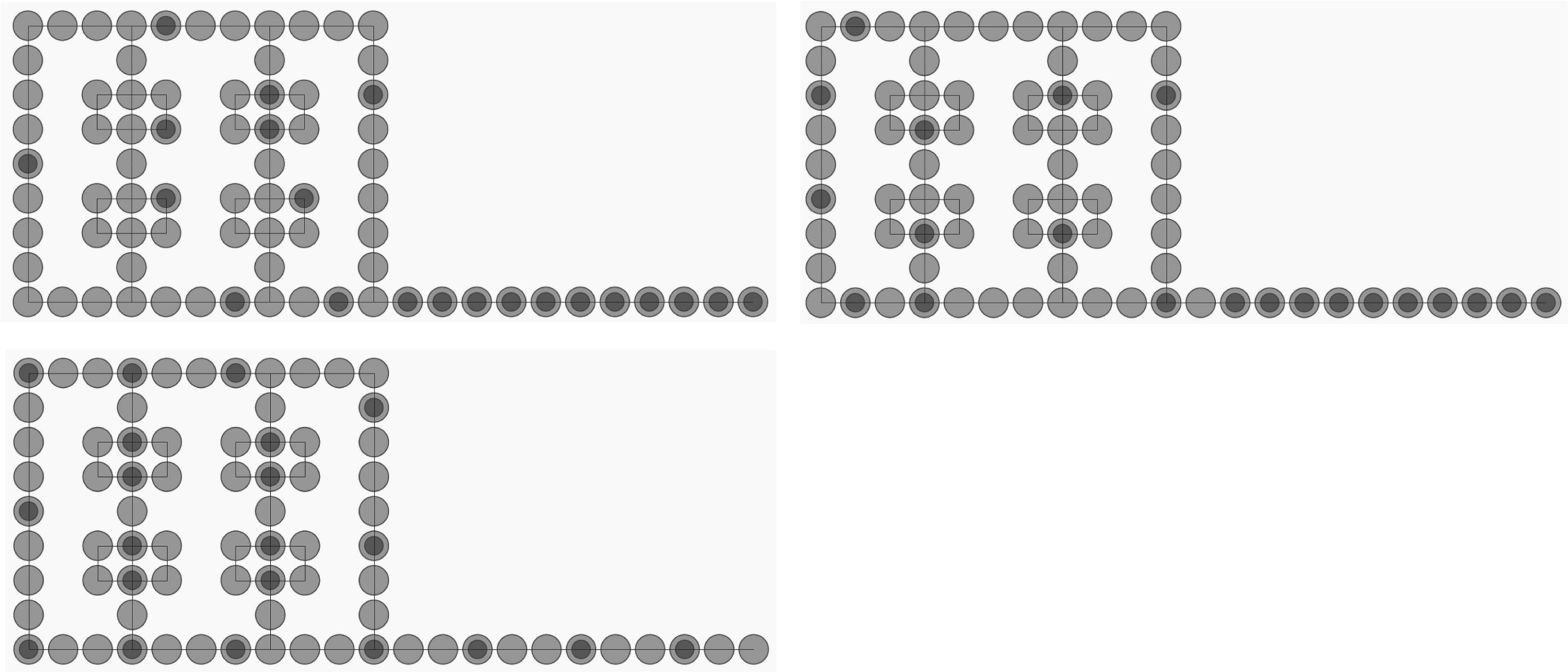}
\vspace{-0.3cm}
\caption{Setting for indoor applications with $N=21$. The space 
consists of narrow corridors and small rooms. Top left: initial allocation, Bottom left: proposed method, Top right: benchmark \cite{sadeghi2020approximation}.}
\label{fig:indoor}
\end{figure}

\begin{figure}[t]
\centering
\includegraphics[width=1.\linewidth]{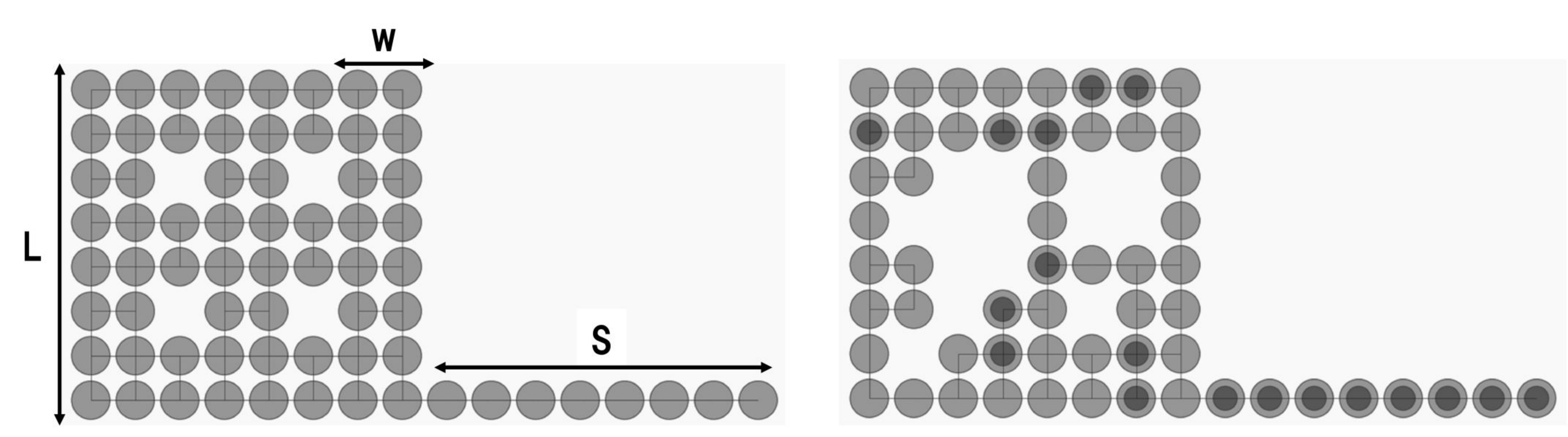}
\vspace{-0.3cm}
\caption{Randomly generated structures. Left: the template structure with the parameter of corridor width $w$, with $L=3(w+1)-1$, and the length of the tail $S=2w+4$. Right: connected structure generated by randomly removing nodes.}
\label{fig:randmaze}
\end{figure}

\begin{figure}[t]
\centering
\includegraphics[width=0.9\linewidth]{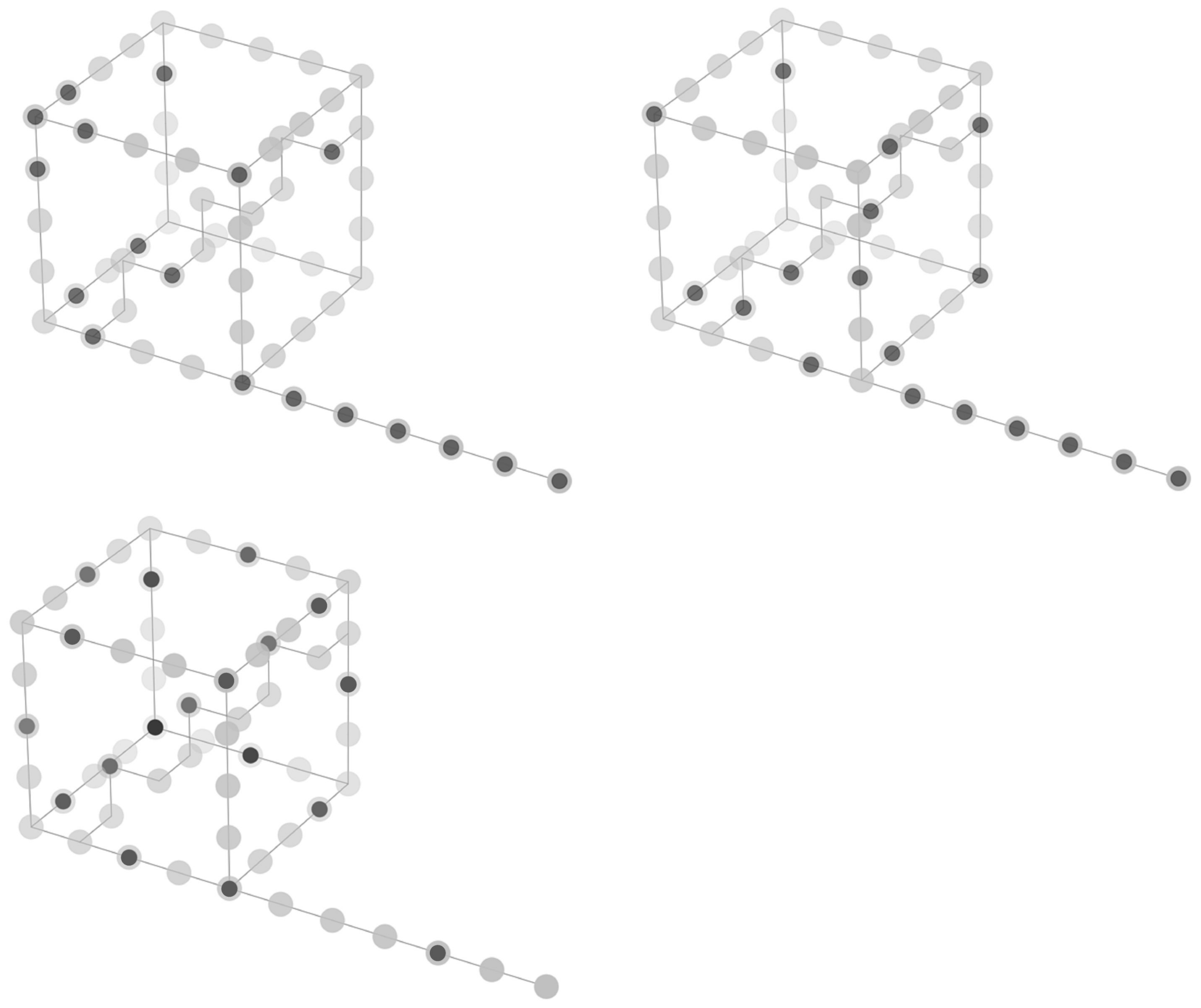}
\vspace{-0.3cm}
\caption{Setting for 3D application with $N=18$. The space 
consists of narrow corridors. Top left: initial allocation, Bottom left: proposed method, Top right: benchmark \cite{sadeghi2020approximation}.}
\label{fig:3d}
\end{figure}



We evaluate the scalability of the proposed algorithm by changing the size of the space $|C|$ and the number of agents $n$. Also, we set $|\mc{C}_+|=|C|$ (the top figure in Figure \ref{fig:runtime}). The middle figure shows the runtime until the convergence with different $|C|$. The growing speed of the runtime closely matches the theoretical prediction of $\mc{O}(|\mc{P}_S|^2)$, which is polynomial. 
The bottom figure shows the runtime with different $n$. Surprisingly, the runtime decreases when $n$ increases, because the size of the partitions $|\mc{P}_S|$ also decreases.
Note that the current implementation uses a single CPU just for theoretical verification, and parallel computation for each agent can reduce the runtime further. 

\begin{figure}[t]
\centering
\includegraphics[width=0.4\linewidth]{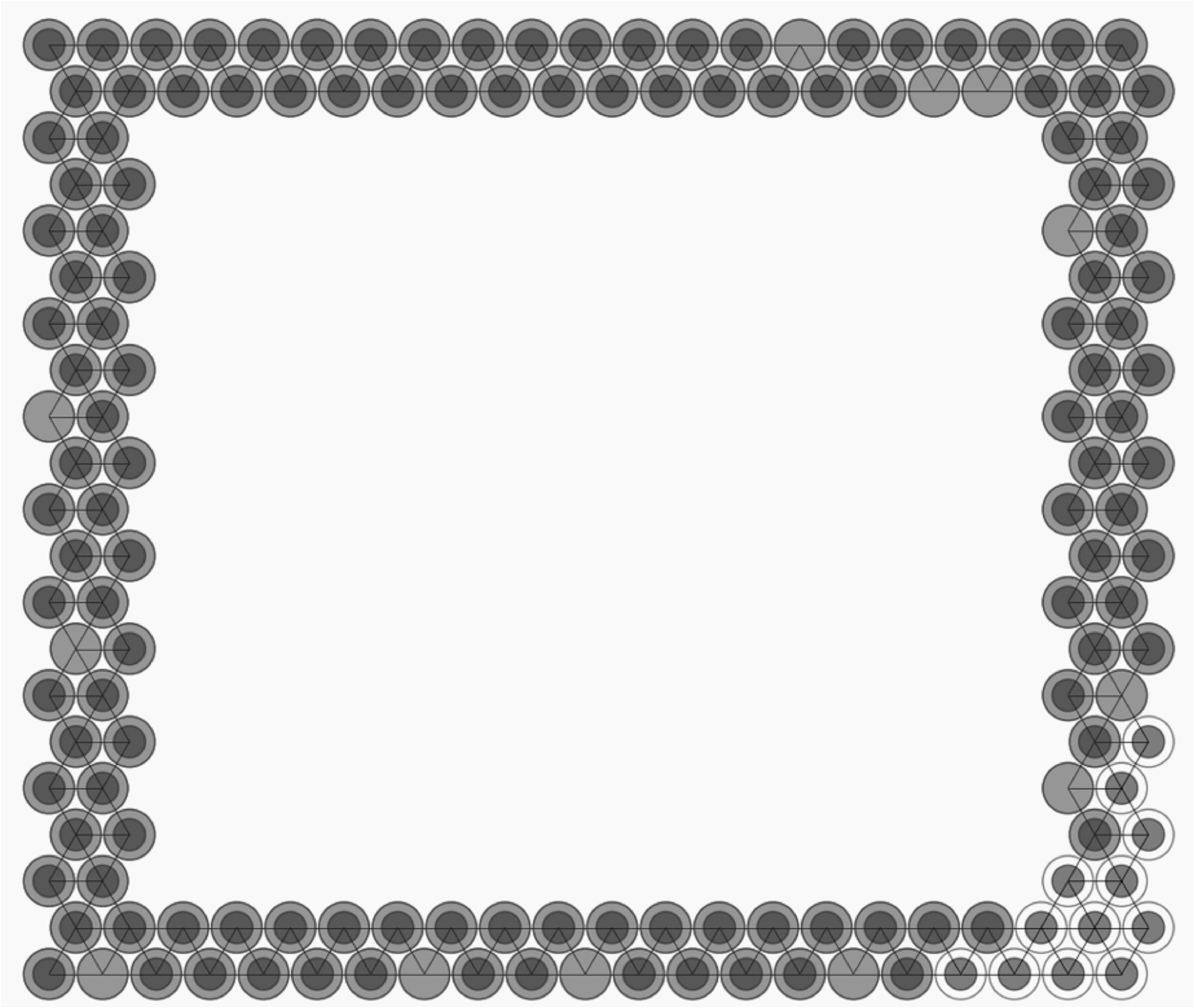}
\includegraphics[width=.6\linewidth]{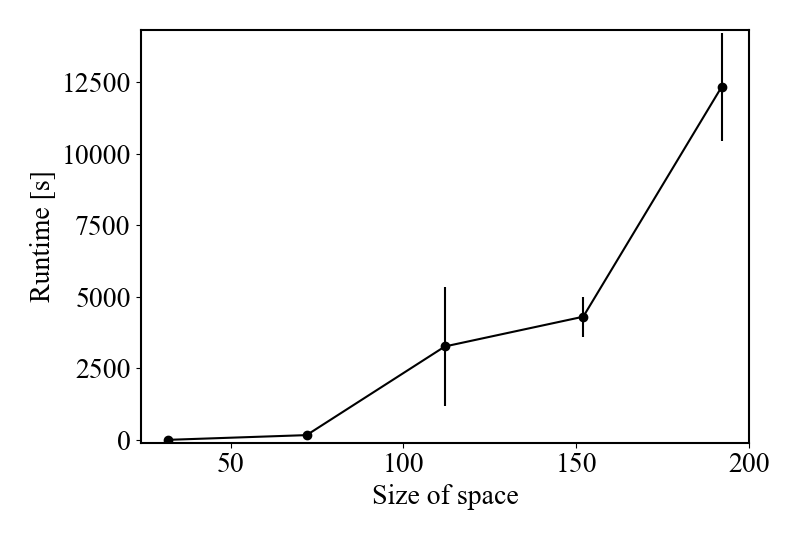}
\includegraphics[width=.6\linewidth]{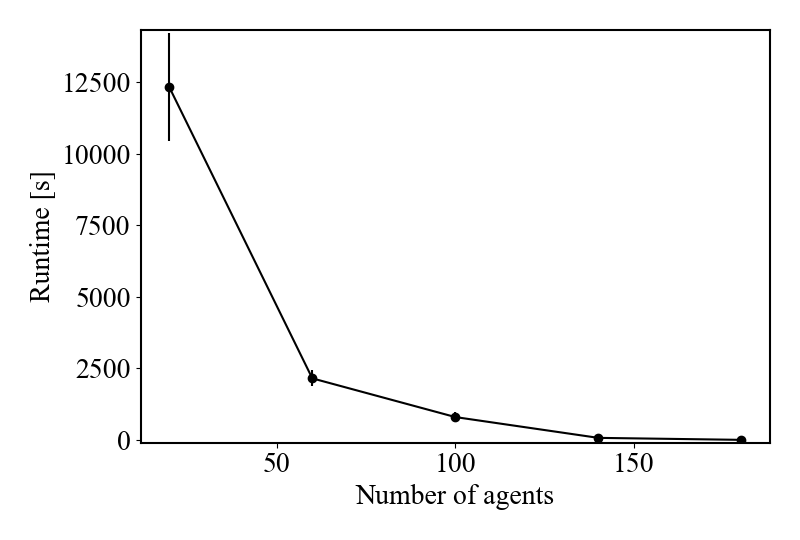}
\caption{Up: The environment ($|C|=152$) to evaluate the scalability. Middle: Runtime of our algorithm (Section \ref{sec:scale}) until the convergence when chainging $|C|$ ($n=20$). Bottom: Runtime of our algorithm when chainging $n$ ($|C|=192$).}
\label{fig:runtime}
\end{figure}
\color{black}

\section{Reproducibility Checklist}
\label{sec:apprepro}

\subsection{Algorithm}

If the paper introduces a new algorithm, it must include a conceptual outline and/or pseudocode of the algorithm for the paper to be classified as CONVINCING or CREDIBLE. (CONVINCING)

\subsection{Theoretical contribution}

If the paper makes a theoretical contribution: 
\begin{enumerate}
    \item All assumptions and restrictions are stated clearly and formally (yes)
    \item All novel claims are stated formally (e.g., in theorem statements) (yes)
    \item Appropriate citations to theoretical tools used are given (yes)
    \item Proof sketches or intuitions are given for complex and/or novel results (yes)
    \item Proofs of all novel claims are included (yes)
\end{enumerate}

For a paper to be classified as CREDIBLE or better, we expect that at least 1. and 2. can be answered affirmatively, for CONVINCING, all 5 should be answered with YES. (CONVINCING)

\subsection{Data sets}

If the paper relies on one or more data sets: 

\begin{enumerate}
    \item All novel datasets introduced in this paper are included in a data appendix (NA)
    \item All novel datasets introduced in this paper will be made publicly available upon publication of the paper (NA)
    \item All datasets drawn from the existing literature (potentially including authors’ own previously published work) are accompanied by appropriate citations (NA)
    \item All datasets drawn from the existing literature (potentially including authors’ own previously published work) are publicly available (NA)
    \item All datasets that are not publicly available (especially proprietary datasets) are described in detail (NA)
\end{enumerate}

Papers can be qualified as CREDIBLE if at least 3., 4,. and 5,. can be answered affirmatively, CONVINCING if all points can be answered with YES. (NA)

\subsection{Experiments}

If the paper includes computational experiments:

\begin{enumerate}
    \item All code required for conducting experiments is included in a code appendix (yes)
    \item All code required for conducting experiments will be made publicly available upon publication of the paper (yes)
    \item Some code required for conducting experiments cannot be made available because of reasons reported in the paper or the appendix (NA)
    \item This paper states the number and range of values tried per (hyper-)parameter during development of the paper, along with the criterion used for selecting the final parameter setting (yes)
    \item This paper lists all final (hyper-)parameters used for each model/algorithm in the experiments reported in the paper (yes)
    \item In the case of run-time critical experiments, the paper clearly describes the computing infrastructure in which they have been obtained (yes)
\end{enumerate}

For CREDIBLE reproducibility, we expect that sufficient details about the experimental setup are provided, so that the experiments can be repeated provided algorithm and data availability (3., 5., 6.), for CONVINCING reproducibility, we also expect that not only the final results but also the experimental environment in which these results have been obtained is accessible (1., 2., 4.). (CONVINCING)

\end{document}